\documentclass[12pt]{article}

\newcommand{\blind}{0}

\usepackage{xcolor, soul}
\usepackage{float}
\usepackage{natbib,verbatim}
\usepackage{amsmath}
\usepackage{amsthm}
\usepackage{amssymb}
\usepackage{enumerate}
\usepackage{setspace}
\usepackage{mathtools}
\usepackage[colorlinks=true, linkcolor=red, urlcolor=blue, citecolor=blue, anchorcolor=blue]{hyperref}       % hyperlinks ,backref=page
\usepackage{todonotes} % \todo{ Text } adds nice todo to compiled version
\usepackage{cleveref} % \cref auto includes reference object type
\usepackage[export]{adjustbox}
\usepackage{algorithm}
\usepackage{algpseudocode}

\crefname{figure}{Figure}{Figures}
\crefname{equation}{}{}
\crefname{definition}{Definition}{Definitions}
\crefname{corollary}{Corollary}{Corollaries}
\crefname{proposition}{Proposition}{Propositions}
\crefname{theorem}{Theorem}{Theorems}
\crefname{remark}{Remark}{Remarks}
\crefname{principle}{Principle}{Principles}
\crefname{lemma}{Lemma}{Lemmata}
\crefname{claim}{Claim}{Claims}
\crefname{table}{Table}{Tables}
\crefname{section}{Section}{Sections}
\crefname{subsection}{Section}{Sections}
\crefname{subsubsection}{Section}{Section}
\crefname{assumption}{Assumption}{Assumptions}
\crefname{appendix}{Appendix}{Appendices}
\numberwithin{equation}{section}
\usepackage{notation}
\numberwithin{theorem}{section}
\usepackage{soul}

\PassOptionsToPackage{numbers, compress}{natbib}

\begin{document}

\def\spacingset#1{\renewcommand{\baselinestretch}{#1}\small\normalsize} \spacingset{1}

% \date{}

\if0\blind
{
  \title{\bf Inference on the proportion of variance \\explained in principal component analysis}
  \author{Ronan Perry\thanks{
    The authors acknowledge funding from the following sources: NSF DMS 2322920, NSF DMS 2514344, NIH 5P30DA048736, ONR N00014-23-1-2589, and a Simons Investigator Award in Mathematical Modeling of Living Systems to D. Witten. NSF DMS 1951980, DMS 2113342 and NSF CAREER Award DMS-2337882 to S. Panigrahi.
    }\\
    Department of Statistics, University of Washington\\
    and \\
    Snigdha Panigrahi \\
    Department of Statistics, University of Michigan\\
    and \\
    Jacob Bien \\
    Department of Data Sciences and Operations, University of Southern California\\
    and \\
    Daniela Witten \\
    Departments of Statistics and Biostatistics, University of Washington
    }
  \maketitle
} \fi

\if1\blind
{
  \bigskip
  \bigskip
  \bigskip
  \begin{center}
    {\Large\bf Inference on the proportion of variance explained\\ in  principal component analysis}
\end{center}
  \medskip
} \fi

\bigskip

\begin{abstract}
Principal component analysis  (PCA) is a longstanding approach for dimension reduction. It rests upon the assumption that the underlying signal has low rank, and thus can be well-summarized using a small number of dimensions. The output of PCA is typically represented using a scree plot, which displays the proportion of variance explained (PVE) by each principal component. While the PVE is extensively reported in routine analyses, to the best of our knowledge the notion of \emph{inference} on the PVE remains unexplored. 

We consider inference on a new population quantity for the PVE with respect to an unknown matrix mean. Our interest lies in the PVE of the sample principal components (as opposed to unobserved population principal components); thus, the population PVE that we introduce is defined \emph{conditional} on the sample singular vectors.
We show that it is possible to conduct inference, in the sense of confidence intervals, p-values, and point estimates, on this population quantity. Furthermore, we can conduct valid inference on the PVE of a subset of the principal components, even when the subset is selected using a data-driven approach such as the elbow rule. We demonstrate our approach in simulation and in an application to gene expression data.
\end{abstract}

\noindent%
{\it Keywords:} Selective Inference, Principal Component Analysis, Proportion of Variance Explained, Scree Plot, Elbow Rule

\if1\blind
{
\vfill

\newpage
} \fi

\spacingset{1.9} % DON'T change the spacing!

\section{Introduction}

Principal component analysis (PCA)
is a longstanding and well-studied technique for dimension reduction, with applications across disciplines~\citep{kaiser_application_1960, wold_cross-validatory_1978, muirhead_aspects_1982, krzanowski_cross-validation_1987, besse_application_1993, bishop_bayesian_1998, bishop_pca_1999,  tipping_mixtures_1999, jolliffe_principal_2002}. We begin with an $n \times p$ data matrix $X = \begin{pmatrix} X_1 & \ldots & X_p \end{pmatrix}$, consisting of $n$ observations and $p$ features.  
 The first principal component loading vector, $v_1 \in \mathbb{R}^p$, is the unit vector ($\| v_1 \|_2=1)$ that maximizes the sample variance of $Xv_1$.
The second principal component loading vector, $v_2  \in \mathbb{R}^p$, is the unit vector that maximizes the sample variance $Xv_2$, subject to the constraint that $v_1^\top v_2=0$, and so on. 
Once the first $k < \min(n,p)$ principal component loading vectors have been computed, they can serve a number of purposes:
\begin{enumerate}
    \item They can be used for visualization:  rather than plotting the $n$ observations with respect to the original set of features $X_1,\ldots,X_p$, we instead plot $X v_1, X v_2,\ldots, X v_k$~\citep{margulies_situating_2016, becht_dimensionality_2019}.
   \item They can be used as features in an  unsupervised analysis, such as clustering. That is, we can cluster the rows of the  $n \times k$ data matrix $\begin{pmatrix} X v_1,\ldots,X v_k \end{pmatrix}$, rather than the rows of the original $n \times p$ data matrix $X$~\citep{calhoun_method_2001,von_luxburg_tutorial_2007, athreya_statistical_2018}.
        \item 
 They can be used to construct  features for a supervised analysis.  For instance, in principal components regression, we predict an $n$-dimensional response $Y$ on the basis of $X v_1,\ldots,X v_k$, as opposed to predicting $Y$ using the original features $X_1,\ldots,X_p$~\citep{mcgough_modeling_2020, hayes_estimating_2025}.
\end{enumerate}

In each of these cases, we implicitly assume that the first few principal component loading vectors capture the signal in the data, and the  remaining dimensions contain mostly noise.
 The suitability of this assumption on a particular dataset is typically visualized using a scree plot,  which displays  $\sval{X}_l^2 / \sum_{l'=1}^p \sval{X}_{l'}^2$ for $l=1,\ldots,p$, where $\sval{X}_1,\ldots,\sval{X}_p$ represent the singular values of $X$ in decreasing order. This quantity is referred to as the ``proportion of variance explained" (PVE) by the $l$th principal component, and it is often reported in scientific publications~\citep{alter_singular_2000, 
 lukk_global_2010, li_dimension_2010, duforet-frebourg_detecting_2016}. In fact, in many cases, the PVE is the \emph{only} quantitative summary of PCA reported. 
In this paper, we consider inference on the PVE, a problem that --- to our knowledge --- has not been previously explored.  We develop a  ``population" quantity corresponding to the PVE, as well as tools to construct confidence intervals and conduct hypothesis tests for this quantity. 

 While in some cases the PVE of all 
 $\min(n,p)$ principal components may be of interest, in others interest may focus on a small number of principal components, where this number is selected from the data using a formal or heuristic approach. In fact, this number is often selected by examining the scree plot, and looking for an ``elbow" where the  PVE decreases \citep{cattell_scree_1966,  li_dimension_2010, senbabaoglu_critical_2014,  nguyen_ten_2019,  greenacre_principal_2022}. As demonstrated in~\cref{fig:toy_example}, if the number of principal components to be considered is selected from the data, then valid inference for the PVE must account for the selection process; thus, we extend  our approach for inference on the PVE to account for the possibility of data-driven selection of the number of principal components. 
 
Our work builds upon the proposal of \cite{choi_selecting_2017}, which considers inference on the singular values of the data. However, they do not consider the PVE, which is  a quantity of much greater scientific interest, and certainly much more widely-reported, than the singular values themselves. Furthermore, they do not account for the possibility that the number of principal components under consideration was selected from the data. 

We conduct inference on the PVE conditional on the $k$th sample singular vectors. Furthermore, in the setting where the number of principal components under consideration is selected from the data, we condition also on this selection event. The selective inference framework~\citep{fithian_optimal_2017} has been used to advantage in a number of settings involving data-driven selection events, including high-dimensional regression~\citep{lee_exact_2016, panigrahi2023approximate}, hierarchical and $k$-means clustering~\citep{gao_selective_2022, chen_selective_2023, yun_selective_2023}, decision trees~\citep{neufeld_tree_2022}, outlier detection~\citep{chen_valid_2020}, and changepoint detection~\citep{ hyun_exact_2018, hyun_post-selection_2021, jewell_testing_2022}. However, our application of this framework to the setting of inference on the PVE is entirely novel. 

\begin{figure}[!hb]
    \centering
    \includegraphics[width=\textwidth]{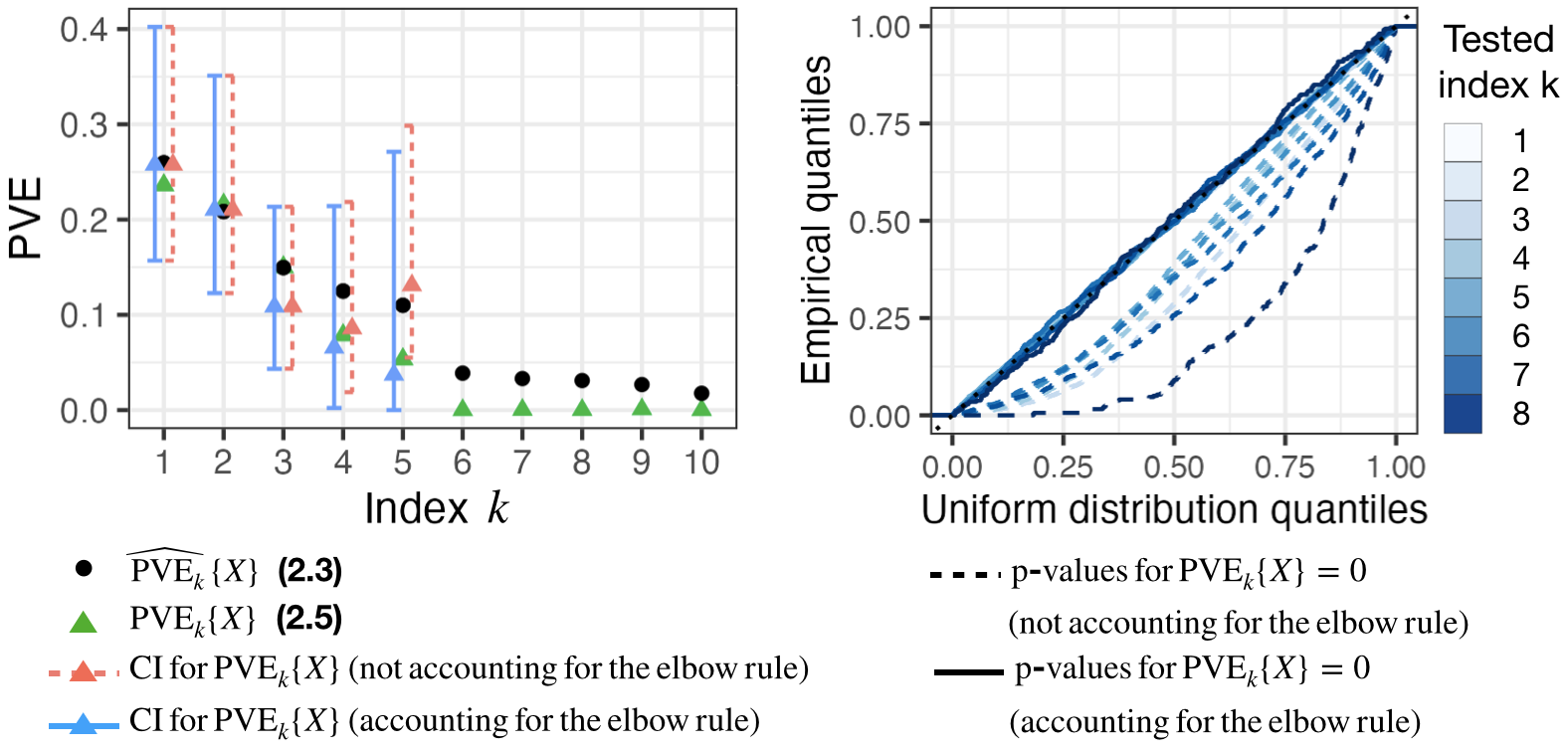}
    \caption{
        \emph{(Left):} The scree plot displays  the sample PVE from~\eqref{eq:pve} (black dots), as well as the true population $\PVE{X}_k$ defined in~\eqref{eq:pop-pve} (green triangles). The elbow rule from~\eqref{equation:zg_rule} selects five principal components. Confidence intervals for $\PVE{X}_k$~\cref{eq:pop-pve} and corresponding point estimates of $\PVE{X}_k$~\cref{eq:mle_pve}, obtained with and without accounting for selection via the elbow rule (blue and red triangles), are plotted. The confidence interval that does not account for selection using the elbow rule fails to cover $\PVE{X}_k$ when $k=5$, and the point estimate is high; accounting for selection mitigates this issue.
        \emph{(Right):} We simulated data for which $\PVE{X}_k = 0$ for all $k$ and $X$, i.e., the population PVE is zero. We then tested $H_{0,k}(X) : \PVE{X}_k = 0$ for all $k \leq r(\sval{X})$, the index  selected by the elbow rule in \eqref{equation:elbow_rule}.  The qq-plot displays the distribution of the resulting p-values across $10^4$ realizations. For example, the curve corresponding to $k=5$ represents the distribution of all p-values that result from testing $H_{0,5}(X)$, which is tested whenever $r(\sval{X}) \geq 5$. The p-values that account for selection using the elbow rule control the type 1 error, whereas those that do not account for selection do not, especially as larger indices are tested.
    }
    \label{fig:toy_example}
\end{figure}

In this paper, we fill a major gap involving inference on the PVE, a quantity that is of critical importance in applied settings. 
 We enable valid inference  for a data analysis pipeline that is  ubiquitous across virtually all areas of data-driven inquiry, but that has lacked solid methodological or theoretical footing.  
Our  primary contributions are as follows: 
\begin{enumerate}
    \item We develop a notion of a ``population quantity" corresponding to PVE. 
    \item We develop hypothesis tests of this population quantity that control the type 1 error, confidence intervals that attain the nominal coverage, and point estimates.
    \item We extend these inferential tools to the setting where the number of principal components considered was selected using a data-driven rule $r(\sval{X})$ based on the singular values $\sval{X}$, e.g., by looking for an ``elbow" in the scree plot.  
    \item We develop efficient strategies to compute the aforementioned quantities.
\end{enumerate}

Figure~\ref{fig:toy_example} illustrates our paper's key contributions on a toy example. In the left-hand panel, we consider a data matrix $X$ for which the ``elbow rule" selects the first five principal components, i.e., $r(\sval{X}) = 5$. For each principal component, we display both the sample PVE \eqref{eq:pve} and a  new population quantity, $\PVE{X}_k$, that we will define in \eqref{eq:pop-pve}. We also display confidence intervals  for $\PVE{X}_1,\ldots,\PVE{X}_5$, obtained using ideas from Section~\ref{sec:confidence_interval}. The solid intervals account for the fact that the elbow rule selected five principal components, whereas the dashed ones do not. There is a difference between the two sets of intervals: in particular, the interval for $\PVE{X}_5$ that does not account for application of the elbow rule does not cover the population quantity. Finally, on the right-hand side, we show that the test of $H_{0,k}: \PVE{X}_k=0$ developed in Section~\ref{sec:hypothesis_testing} leads to type 1 error control, provided that we account for the elbow rule. In both panels, \emph{failure to account for selection using the elbow rule is most severe when we conduct inference on $\PVE{X}_{k}$ for $k = r(\sval{X})$}; in the right panel, this is evidenced by the fact that p-values are most anti-conservative at index $k=8$, which is only ever tested when $r(\sval{X}) = 8$ since we empirically never observe $r(\sval{X}) > 8$. By contrast, when we conduct inference on $\PVE{X}_k$ for $k < r(\sval{X})$, then failure to account for the elbow rule has relatively little effect.

The rest of our paper is organized as follows. In~\cref{sec:setup}, we present a model for PCA, and derive the density of the singular values, conditional on a selection event based on an arbitrary data-driven rule.  
We further define a population quantity corresponding to the PVE. 
In~\cref{sec:inference}, we develop hypothesis tests and confidence intervals corresponding to this population notion of PVE, which we further extend to allow for a generic data-driven selection of the principal components considered for inference. 
Efficient computation of these p-values and confidence intervals, in the case where selection is performed using the elbow rule, is discussed in \cref{sec:computation}. Numerical results are presented in~\cref{sec:simulation_study}. We close with a discussion in Section~\ref{sec:discussion}. Proofs of all technical results, as well as additional numerical results, can be found in the Appendix. 
 
\section{Preliminaries for inference on the PVE}\label{sec:setup}

\subsection{A model for PCA}\label{subsec:pca_model}

A number of models for PCA have been considered in the literature. The ``probabilistic PCA" approach of  \cite{bishop_pca_1999} and \cite{minka_automatic_2000}  assumes that each  row of the $n \times p$ data matrix $X$ is independently drawn from a multivariate normal distribution with mean zero and a spiked covariance matrix. Here, instead, we consider a  model for PCA in which the low-rank structure is contained in the mean, rather than in the covariance matrix~\citep{muirhead_aspects_1982, choi_selecting_2017}. That is, we assume that 
    $X = \Theta + E$, 
where $\Theta \in \RR^{n \times p}$ is a fixed (non-random) mean matrix, and $E \in \RR^{n \times p}$ is a noise matrix with independent and identically-distributed elements, $E_{ij} \sim \Norm(0, \sigma^2)$. Equivalently, 
 \begin{equation}\label{equation:mvn_model}
    X \sim \text{MN}(\Theta, I_n, \sigma^2 I_p),
\end{equation} 
a matrix normal random variable with mean $\Theta$, independent rows, and columns with covariance matrix $\sigma^2 I_p$. Without loss of generality, we assume that $n \geq p$.
 
 Before performing PCA, it is typical to column-center the data matrix. In Appendix~\ref{appendix:subsec:centering}, we show that column-centering amounts to a slight re-parametrization of \eqref{equation:mvn_model}; thus, the discussion that follows applies with or without column-centering.

  The singular value decomposition of $X$ can be written as 
\begin{equation}
    X = \Usvd{X}
    \begin{bmatrix}
        \sval{X}_1  & \dots & 0\\
        \vdots & &\vdots\\
        0 & \dots & \sval{X}_p
    \end{bmatrix}
    \Vsvd{X}^\top.
     \label{equation:svd}
\end{equation}
Here, $\Usvd{\cdot}$ and $\Vsvd{\cdot}$ are functions mapping to $n \times p$ and $p \times p$ matrices, respectively, with orthonormal columns: that is, $\Usvd{X}^\top \Usvd{X} = \Vsvd{X}^\top \Vsvd{X} = I_{p \times p}$.   
The columns of $\Vsvd{X}=\begin{bmatrix} \Vsvd{X}_1 \ldots \Vsvd{X}_p \end{bmatrix}$ are the principal component loading vectors of $X$, while those of $\Usvd{X}=\begin{bmatrix} \Usvd{X}_1 \ldots \Usvd{X}_p \end{bmatrix}$ are the (normalized) principal component score vectors. We can similarly decompose the unknown mean matrix $\Theta$, defining $\Usvd{\Theta}=\begin{bmatrix} \Usvd{\Theta}_1 \ldots \Usvd{\Theta}_p \end{bmatrix}$, $\Vsvd{\Theta}=\begin{bmatrix} \Vsvd{\Theta}_1 \ldots \Vsvd{\Theta}_p \end{bmatrix}$, and $\sval{\Theta} = \rbr{\sval{\Theta}_1, \dots, \sval{\Theta}_p}$.

\subsection{A population parameter for the PVE}
The sample PVE~\citep{jolliffe_principal_2002, li_dimension_2010} of the $k$th principal component of the data matrix $X$ is defined as
\begin{equation}\label{eq:pve}
    \PVEhat{X}_k = \frac{\eval{X}_k}{\sum_{k'=1}^p \eval{X}_{k'}}.
    % = \frac{\eval{X}_k}{\nbr{X}_F^2}.
\end{equation}
We will now define a  population parameter corresponding to this sample quantity. It is tempting to consider
\begin{equation}
\frac{\eval{\Theta}_k}{\sum_{k'=1}^p \eval{\Theta}_{k'}}
= \frac{  \left( \Usvd{\Theta}_k^\top \Theta \Vsvd{\Theta}_k \right)^2 }{\nbr{\Theta}_F^2} 
\label{eq:pop-pve-wrong}
\end{equation}
as a candidate for this population parameter. However, interest lies in the variance explained by the $k$th \emph{sample} singular vectors,  $\Usvd{X}_k$ and $\Vsvd{X}_k$, not the variance explained by the $k$th \emph{population} singular vectors, $\Usvd{\Theta}_k$ and $\Vsvd{\Theta}_k$: this is because the former are part of the data analysis pipeline and downstream tasks, whereas the latter are not.
Thus, we replace $\Usvd{\Theta}_k$ and $\Vsvd{\Theta}_k$ in~\cref{eq:pve} with $\Usvd{X}_k$ and $\Vsvd{X}_k$.

This results in the estimand 
\begin{equation}
  \PVE{X}_k =   \frac{ \left( \Usvd{X}_k^\top \Theta \Vsvd{X}_k \right)^2 }{\nbr{\Theta}_F^2}.
    \label{eq:pop-pve}
\end{equation}
We will view \eqref{eq:pop-pve} as the population quantity corresponding to $\PVEhat{X}_k$ \eqref{eq:pve} in the rest of this paper. Note that under the \emph{global null} in which $\nbr{\Theta}_F^2 = 0$, we define $\PVE{X}_k = 0$. We interpret \cref{eq:pop-pve} as the amount of information that the $k$th singular vectors of $X$ contain about the mean $\Theta$, relative to the total amount of signal contained in $\Theta$. By Cauchy-Schwarz, $\sum_{k=1}^{p} \PVE{X}_k \leq 1$, with equality holding if the sample singular vectors are also population singular vectors. By contrast, of course $\sum_{k=1}^p \PVEhat{X}_k = 1$.

We emphasize that --- as suggested by the unusual notation involving $X$ --- the quantity $\PVE{X}_k$ in \eqref{eq:pop-pve} is not a ``conventional" parameter: it involves not only the unknown mean parameter $\Theta$ but also the data $X$ (because $\Usvd{X}_k$ and $\Vsvd{X}_k$ are derived from the data).  To conduct inference on $\PVE{X}_k$, we will need to condition on $\Usvd{X}_k$ and $\Vsvd{X}_k$. We discuss this in greater depth in~\cref{sec:density}.

\subsection{Deriving the conditional density}
 \label{sec:density}

We now lay the foundation for developing hypothesis tests in~\cref{sec:hypothesis_testing}, confidence intervals in~\cref{sec:confidence_interval}, and point estimates in~\cref{sec:mle}. Lemma 1 from~\citet{choi_selecting_2017} uses results from~\citet[p. 105]{muirhead_aspects_1982} to establish that under the data generating process in \eqref{equation:mvn_model}, the joint density of $\sval{X}, \Usvd{X}$, and $\Vsvd{X}$ takes the form
\begin{align}
    \label{equation:joint_density}
    f\rbr{\sval{X}, \Usvd{X}, \Vsvd{X}}
        &\propto \exp\rbr{
            -\frac{1}{2\sigma^2} \sbr{ \sum_{k=1}^p \sval{X}_k^2 -2 \sum_{k=1}^p \sval{X}_k \Usvd{X}_k^\top \Theta\Vsvd{X}_k + \tr(\Theta^\top \Theta) }
        } \nonumber\\
        &\quad \times J(\sval{X}) \,d\mu_{n,p}\rbr{\Usvd{X}} \, d\mu_{p,p}\rbr{\Vsvd{X}} \nonumber\\
        &\quad \times \prod_{k=1}^{p} \II_{[\sval{X}_{k+1}, \sval{X}_{k-1}]}(\sval{X}_k),
\end{align}
where we define
\begin{equation}
    \label{equation:J(s)}
   J(\sval{X}) := \prod_{k=1}^p \sval{X}_k^{n-p} \prod_{1 \leq k < k' \leq p}\rbr{\eval{X}_k - \eval{X}_{k'}},
\end{equation}
letting $\sval{X}_0 := \infty$, $\sval{X}_{p+1}:= 0$, and $\mu_{a,b}(\cdot)$ denote the uniform measure over matrices in $\RR^{a \times b}$ with orthogonal columns.

Recall that we wish to conduct inference on $\PVE{X}_k$ defined in \eqref{eq:pop-pve}. We make the following observations:
\begin{enumerate}
    \item  $\PVE{X}_k$ involves the data by way of $\Usvd{X}_k$ and $\Vsvd{X}_k$. Thus, we will condition on $\Usvd{X}_k$ and $\Vsvd{X}_k$ so that $\PVE{X}_k $ is deterministic.
    \item The numerator of $\PVE{X}_k$ depends on the unknown parameter $\Theta$ only by way of $\pivot{X}_k$. This quantity is coupled with the $k$th singular value $\sval{X}_k$ in the joint density in \eqref{equation:joint_density}. Since \eqref{equation:joint_density} also involves  $\Usvd{X}_{k'}$, $\Vsvd{X}_{k'}$, and $\sval{X}_{k'}$ for $k'\neq k$, we will further condition on these quantities.
    \item Typically, a data analyst is  interested in the PVE of the $k$th principal component only if  $k \leq r(\sval{X})$, where $r(\sval{X})$ is the number of principal components selected using a data-driven approach based on the singular values $\sval{X}$. Thus, we condition on the \textit{selection event} $k \leq r(\sval{X})$. In the absence of a data-driven selection approach, we define $r(\sval{X}) := p$. 
\end{enumerate}

We therefore need to characterize the distribution of $\sval{X}_k$ given the other singular values, $\Usvd{X}$, $\Vsvd{X}$, and the selection event $k \leq r(\sval{X})$.

\begin{proposition}[The conditional density of $\sval{X}_k$]\label{lemma:si_density}
    Define
    \begin{equation*}
        \sval{X}_{[-k]} := \rbr{\sval{X}_1, \dots, \sval{X}_{k-1}, \sval{X}_{k+1}, \dots, \sval{X}_p},
    \end{equation*}
    Under~\eqref{equation:mvn_model}, the 
     density of  $\sval{X}_k$, conditional on the values of $\Usvd{X}$, $\Vsvd{X}$, $\sval{X}_{[-k]}$, and the selection event $k \leq r(\sval{X})$, is
    \begin{align}\label{equation:si_density}
                &f \left(\sval{X}_k~|~\Usvd{X}, \Vsvd{X}, \sval{X}_{[-k]}, k \leq r(\sval{X}) \right)\nonumber\\  &\propto
          \exp\bigg(
                -\frac{1}{2\sigma^2} \sval{X}_k^2  +
                \frac{1}{\sigma^2} \sval{X}_k \Usvd{X}_k^\top \Theta \Vsvd{X}_k
                \bigg)
          \times \sval{X}_k^{n-p} \times \Pi_{k' \in [p]: k' \neq k} \left|\sval{X}_k^2 - \sval{X}_{k'}^2 \right|\nonumber \\
          &\quad\times \II_{[\sval{X}_{k+1}, \sval{X}_{k-1}]}(\sval{X}_k)
          \times \II[k \leq r(\sval{X})],
    \end{align}
    where $\sval{X}_0 := \infty$, $\sval{X}_{p+1}:= 0$.
    \end{proposition}
In~\cref{lemma:si_density}, $\II_{A}(\cdot)$ denotes an indicator function that equals  one if its argument is in the set $A$, and zero otherwise.

\begin{proposition}\label{prop:theta_delta_k}
    Let $x \in \RR^{n \times p}$, and define
    \begin{equation}
        \label{equation:s_impute}
        \simpute{x}{t}^k := (\sval{x}_1, \dots, \sval{x}_{k-1}, t, \sval{x}_{k+1}, \dots, \sval{x}_p),
    \end{equation}
    \begin{equation}
        \label{equation:h}
        h\left(t; \delta, \sval{x}_{[-k]} \right) := \exp{\left(-\frac{t^2}{2\sigma^2} + \frac{t}{\sigma^2}\delta\right)} t^{n-p} \prod_{k' \neq k} \left|t^2 - \sval{x}_{k'}^2\right|,
    \end{equation}
    and
    \begin{equation}\label{equation:def-p}
        \PP_{\delta, k}(x) := \frac{
               \int_{\sval{x}_k}^{\sval{x}_{k-1}} h\left(t; \delta, \sval{x}_{[-k]}\right) \II\left[r\left(\simpute{x}{t}^k\right) \geq k \right] dt
        }{
               \int_{\sval{x}_{k+1}}^{\sval{x}_{k-1}} h\left(t; \delta, \sval{x}_{[-k]}\right) \II\left[r\left(\simpute{x}{t}^k\right) \geq k \right] dt }.
    \end{equation}
    Then
    \begin{align*}
        \label{equation:def-p}
         &\prob \big(
        \sval{X}_k \geq \sval{x}_k ~\big|~ \Usvd{X}=\Usvd{x}, \Vsvd{X}=\Vsvd{x},
        \sval{X}_{[-k]}=\sval{x}_{[-k]},
         k \leq r(\sval{X})
         \big)\\
         &= \PP_{\pivot{x}_k, k}(x).
    \end{align*}
       \end{proposition}
Finally, we arrive at this section's main result, which sets the stage for the confidence intervals and hypothesis tests in Section~\ref{sec:inference}.
\begin{theorem}[Conditional distribution of $\PP_{\pivot{x}_k,k}(X)$]\label{thm:uniform}
    Under \eqref{equation:mvn_model}, for any $\alpha \in (0,1)$, $x \in \RR^{n \times p}$, and $k$, it holds that
    \begin{equation*}
       \prob \left( \PP_{\pivot{x}_k,k}(X) \leq \alpha~|~\Usvd{X}_k=\Usvd{x}_k, \Vsvd{X}_k=\Vsvd{x}_k,
         k \leq r(\sval{X}) \right) = \alpha.
    \end{equation*}
 \end{theorem}

\section{Conducting inference on the PVE}
\label{sec:inference}

We now develop hypothesis tests, confidence intervals, and point estimates for $\PVE{x}_k$, for any $x \in \RR^{n \times p}$. All inference in this section will be conditional on
\begin{equation}\label{eq:conditioning_event}
    \widetilde A_k(X, x) := \cbr{\Usvd{X}_k = \Usvd{x}_k, \Vsvd{X}_k=\Vsvd{x}_k, k \leq r\left(\sval{X}\right)}.
\end{equation}
That is, we conduct inference conditional on the $k$th left and right estimated singular vectors, and the fact  that at least $k$  principal components were selected from the data. 

\subsection{Hypothesis tests}
\label{sec:hypothesis_testing}

We consider testing the selective hypothesis $\PVE{X}_k = 0$ in~\eqref{eq:pop-pve}, conditional on $\widetilde A_k(X, x)$ in \eqref{eq:conditioning_event}. Since $\PVE{X}_k = 0$  if and only if $\pivot{X}_k = 0$, the following corollary of \cref{thm:uniform}  guarantees  selective type 1 error control, in the sense of~\citet{fithian_optimal_2017}, for a test based on  $\PP_{0,k}(X)$ defined in \eqref{equation:def-p}.

\begin{corollary}[Uniformity of p-values under the null]\label{cor:pval_uniform}
    For any $\alpha \in (0, 1)$, $x \in \RR^{n \times p}$, the test that rejects $H_{0, k}(x): \PVE{x}_k=0$ if $\PP_{0, k}(X) \leq \alpha$  controls the selective type 1 error at level $\alpha$. That is, under the null hypothesis $H_{0, k}(x): \PVE{x}_k=0$,
    \begin{equation}\label{eq:selT1E}
        \prob\rbr{\PP_{0,k}(X) \leq \alpha ~\big|~\widetilde A_k(X, x)}  = \alpha.
    \end{equation}
\end{corollary}

\begin{remark}[Relationship to \citet{choi_selecting_2017}]
    Comparing the p-value in~\cref{cor:pval_uniform} to that from Theorem 2.1 of~\citet{choi_selecting_2017}, we note that the primary difference is that the latter was computed without conditioning on $k \leq r\left(\sval{X}\right)$, i.e., it does not account for the selection of the number of principal components from the data. 
\end{remark}

\subsection{Confidence intervals}\label{sec:confidence_interval}

We wish to construct a confidence interval for $\PVE{x}_k$ defined in \eqref{eq:pop-pve}, conditional on~\cref{eq:conditioning_event}. That is, our goal is an interval $\text{CI}_{\alpha, k}(X)$ that satisfies
\begin{equation*}
    \Pr\left( \PVE{x}_k \in \text{CI}_{\alpha,k}(X) \mid \tilde{A}_k(X, x) \right) =   \Pr\left( \frac{ \left( \Usvd{x}_k^\top \Theta \Vsvd{x}_k \right)^2 }{\nbr{\Theta}_F^2} \in \text{CI}_{\alpha,k}(X)~\bigg|~ \tilde{A}_k(X, x) \right) \geq 1-\alpha.
\end{equation*}

First, we note that for positive $\alpha_1$, we can obtain a confidence region for the numerator of~\cref{eq:pop-pve} of the form 
\begin{equation*}
    \Pr(\Usvd{x}_k^\top \Theta \Vsvd{x}_k \in \text{C}_{\alpha_1, k}(X) \mid \tilde{A}_k(X, x) ) = 1-\alpha_1
\end{equation*}
by inverting the p-values in~\cref{cor:pval_uniform}, i.e., 
\begin{equation}\label{eq:CI-numer}
  \text{C}_{\alpha_1, k}(X) := \cbr{\delta\in \RR : \PP_{\delta, k}\rbr{X} \in \sbr{\alpha_1/2,\, 1-\alpha_1/2}}.
\end{equation}
More specifically, the conditional density in~\cref{equation:si_density} is a one-parameter exponential family in the parameter $\pivot{X}_k$. By~\cref{prop:theta_delta_k}, $1 - \PP_{\delta, k}\rbr{\cdot}$ is the CDF of this family and hence has a monotone likelihood ratio in $\delta$. Therefore, the confidence region is the interval $\text{C}_{\alpha_1, k}(X) = \sbr{L_{\alpha_1}^{\mathrm{num}}(X), U_{\alpha_1}^{\mathrm{num}}(X)}$, where $L_{\alpha_1}^{\mathrm{num}}(X)$ and $U_{\alpha_1}^{\mathrm{num}}(X)$ satisfy
\begin{equation}
    \PP_{L_{\alpha_1}^{\mathrm{num}}(X), k}(X) = \alpha_1/2,
    \quad\quad
    \PP_{U_{\alpha_1}^{\mathrm{num}}(X), k}(X) = 1 - \alpha_1/2.
\end{equation}

Under a transformation of $\text{C}_{\alpha_1, k}(X)$ (details left to \cref{proof:cor_si_ci_coverage}), we can obtain a $1-\alpha_1$ confidence interval for the square $\rbr{\pivot{x}_k}^2$, denoted $\sbr{\widetilde L_{\alpha_1/2}^{\mathrm{num}}(X), \widetilde U_{\alpha_1/2}^{\mathrm{num}}(X)}$.

Now, suppose that we also had access to an independent dataset $X'$ from our model~\cref{equation:mvn_model}. We could then construct a confidence interval for the denominator of~\cref{eq:pop-pve} of the form 
\begin{equation}
    \Pr\rbr{\|\Theta\|_F^2 \in [L_{\alpha_2}^{\mathrm{denom}}(X'), U_{\alpha_2}^{\mathrm{denom}}(X')]}=1-\alpha_2,
    \label{eq:CI-denom}
\end{equation}
for positive $\alpha_2$. Noting that \eqref{eq:CI-denom} also holds conditional on $\tilde{A}_k(X, x)$,  it would follow that
\begin{align}
    & \Pr\left(\frac{\rbr{\pivot{x}_k}^2}{\|\Theta\|_F^2} \in  \sbr{ \frac{\widetilde L_{\alpha_1}^{\mathrm{num}}(X)}{U_{\alpha_2}^{\mathrm{denom}}(X')}, \frac{\widetilde U_{\alpha_1}^{\mathrm{num}}(X)}{L_{\alpha_2}^{\mathrm{denom}}(X')} } \bigg\vert \tilde{A}_k(X, x) \right)  \nonumber \\
    & \geq  \Pr\Bigg( \left\{ \Usvd{x}_k^\top \Theta \Vsvd{x}_k \in \text{C}_{\alpha_1, k}(X) \right\} \cap \left\{   \|\Theta\|_F^2 \in [L_{\alpha_2}^{\mathrm{denom}}(X'), U_{\alpha_2}^{\mathrm{denom}}(X')]    \right\} \big\vert \tilde{A}_k(X, x) \Bigg) \nonumber \\
    & \geq (1- \alpha_1)(1 - \alpha_2) > 1-\alpha_1 -\alpha_2.
\label{eq:CI-ratio}
\end{align}
We implement this strategy as follows. For a positive constant $c$, we define $X^{(1)} := X + c E'$  for $E' \sim \text{MN}(0, I_n, \sigma^2 I_p)$, a noisy version of  $X$. Similarly, let $x^{(1)} \in \RR^{n \times p}$. Then, rather than a confidence interval for $\PVE{x}_k$ conditional on $\widetilde A_k\rbr{X, x}$, we will instead construct a confidence interval for
\begin{equation}\label{eq:actual-PVE}
    \PVE{x^{(1)}}_k := \frac{\rbr{\pivot{x^{(1)}}_k}^2}{\nbr{\Theta}_F^2}, 
\end{equation}
conditional on $\widetilde A_k\rbr{X^{(1)}, x^{(1)}}$.

A confidence interval for the numerator of \eqref{eq:actual-PVE}, conditional on $\widetilde A_k\rbr{X^{(1)}, x^{(1)}}$, takes a form directly analogous to \eqref{eq:CI-numer}. To obtain a confidence interval for the denominator, let $\sigma^2_c := \sigma^2\rbr{1 + c^{-2}}$ and define $X^{(2)} := X - c^{-1}E' \sim \text{MN}(\Theta, I_n, \sigma^2_c I_p)$. Note that $X^{(1)}$ and $X^{(2)}$ are independent by construction of this data thinning procedure~\citep{rasines_splitting_2023, leiner_data_2023, neufeld_data_2024, dharamshi_generalized_2024}, and $\sigma^{-2}_c \nbr{X^{(2)}}_F^2 \sim \chi^2_{np}\rbr{\sigma^{-2}_c\nbr{\Theta}_F^2}$, i.e., this is a non-central chi-squared random variable.  Recalling that $\sigma$ is known, we can therefore construct an exact $1-\alpha_2$ confidence interval for $\nbr{\Theta}_F^2$ numerically using the quantiles of the non-central $\chi_{np}^2$ distribution. We let $\left[ L_{\alpha_2}^{\mathrm{denom}}\rbr{X^{(2)}}, U_{\alpha_2}^{\mathrm{denom}}\rbr{X^{(2)}} \right]$ denote this interval. 

Thus, for any $\alpha_1, \alpha_2 \in (0,1)$, we define the confidence interval
\begin{align}\label{equation:si_confidence_interval}
    \text{CI}_{\alpha_1, \alpha_2, k}\rbr{X^{(1)}, X^{(2)}}
    := \sbr{ \frac{ \widetilde L_{\alpha_1}^{\mathrm{num}} \rbr{X^{(1)}} }{U_{\alpha_2}^{\mathrm{denom}} \rbr{X^{(2)}}},
    \frac{ \widetilde U_{\alpha_1}^{\mathrm{num}} \rbr{X^{(1)}} }{L_{\alpha_2}^{\mathrm{denom}}\rbr{X^{(2)}}}
    }.
\end{align}

\begin{theorem}[Selective coverage of the confidence interval]\label{thm:si_ci_coverage}
    Let $X$ be distributed according to~\cref{equation:mvn_model}, and for a positive constant $c$, define $X^{(1)} := X + c E'$ and $X^{(2)} := X - c^{-1}E'$ for $E' \sim \text{MN}(0, I_n, \sigma^2 I_p))$. For any $\alpha_1, \alpha_2 \in (0, 1)$ and $x^{(1)} \in \RR^{n \times p}$,
    \begin{equation}
        \prob\rbr{ \PVE{x^{(1)}}_k \in \mathrm{CI}_{\alpha_1, \alpha_2,k}\rbr{X^{(1)}, X^{(2)}} ~\bigg|~\widetilde A_k\rbr{X^{(1)}, x^{(1)}} }  \geq 1-\alpha_1 - \alpha_2.
    \end{equation}
\end{theorem}

\begin{remark}
    Note that $X$ and $X^{(1)}$ have the same mean, but the latter has larger variance. Therefore, the singular vectors  of both $X$ and $X^{(1)}$ provide estimates of the principal components of $\Theta$; however, those arising from $X^{(1)}$ are noisier. Similarly, confidence intervals calculated using $X^{(2)}$ will be wider than those computed using $X$. In short, to provide valid inference after PCA, we accept a reduction in both the accuracy of the estimated singular vectors and the inferential power.
\end{remark}

\begin{remark}
    In practice, we set $c=1$, which equally splits the Fisher information about $\Theta$ into $X^{(1)}$ and $X^{(2)}$. Furthermore, we set $\alpha_1 + \alpha_2 = \alpha$ to control the type I error at a desired level $\alpha$. We balance $\alpha_1$ and $\alpha_2$ to yield empirically narrow average confidence intervals.
\end{remark}

\subsection{Point estimates}
\label{sec:mle}

 To obtain a point estimate for $\PVE{x}_k$ in~\cref{eq:pop-pve} conditional on $\widetilde A_k(X, x)$, we first consider a point estimate for the numerator. To do this, we maximize the likelihood based on the conditional density of $\sval{X}_{k}$ in~\eqref{equation:si_density}. As this is a function of $\Theta$ only through $\pivot{x}_k$, the  point estimate for $\pivot{x}_k$ conditional on $\widetilde A_k(X, x)$ is 
\begin{equation}\label{eq:mle_pivot}
    \hat \delta_{k, \mathrm{MLE}}\rbr{X}
    := \argmax_{\delta \in \RR} \frac{
                h\left(\sval{X}_k; \delta, \sval{X}_{[-k]}\right)
        }{
               \int_{\sval{X}_{k+1}}^{\sval{X}_{k-1}} h\left(t; \delta, \sval{X}_{[-k]}\right) \II\left[r\left(\simpute{X}{t}^k\right) \geq k \right] dt },
\end{equation}
where $h\left(t; \delta, \sval{X}_{[-k]} \right)$ is defined in~\eqref{equation:h}.

We further note that $\nbr{X}_F^2 - np\sigma^2$ is an unbiased estimator of the denominator of~\cref{eq:pop-pve}. However, $\nbr{X}_F^2 - np\sigma^2$ is not unbiased for the denominator of~\cref{eq:pop-pve} conditional on $\tilde{A}_k(X,x)$ in \eqref{eq:conditioning_event}. 
Thus, as in~\cref{sec:confidence_interval},  for a positive constant $c$, we define $X^{(1)} := X + c E'$ and $X^{(2)} := X - c^{-1} E'$ for $E' \sim \text{MN}(0, I_n, \sigma^2 I_p)$; note that $X^{(1)}$ and $X^{(2)}$ are independent \citep{rasines_splitting_2023}.
Our point estimate for $\PVE{x^{(1)}}_k$ conditional on $\widetilde A_k(X^{(1)}, x^{(1)})$ is therefore 
\begin{equation}\label{eq:mle_pve}
   \PVE{X^{(1)}, X^{(2)}}_{k, MLE} := \frac{ \hat\delta_{k, MLE}^2\rbr{X^{(1)}} }{\nbr{X^{(2)}}_F^2 - np\sigma^2_c},
\end{equation}
where $\sigma^2_c := \sigma^2\rbr{1 + c^{-2}}$.

\section{Computation of $\PP_{\delta,k}(X)$}

\label{sec:computation}

We saw in Section~\ref{sec:inference} that the quantity $\PP_{\delta,k}(X)$, defined in \eqref{equation:def-p}, is key to inference on $\PVE{X}_k$. Computing $\PP_{\delta, k}(X)$ in \eqref{equation:def-p} requires evaluating the indicator variable $\II\left[r\left(\simpute{X}{t}^k\right) \geq k \right]$.  
If $r(\sval{X})=p$ --- that is, if we intend to conduct inference on the PVE of all $p$ principal components --- then this indicator always equals one. However, if the function $r(\cdot):\mathbb{R}_{+}^p \rightarrow \{1,\ldots,p\}$ corresponds to a data-driven selection event, then we must characterize the set 
\begin{equation}
    \label{equation:elbow_solution_set}
    \cbr{t: k \leq r\left(\simpute{X}{t}^k\right)}
\end{equation}

in order to compute \eqref{equation:def-p}. 
 This set will of course depend on the form of  $r(\cdot)$. 
 
We now consider  two possible choices of $r(\cdot)$. Both are instantiations of the elbow rule, in the sense that they seek an ``elbow" in the scree plot~\citep{jolliffe_principal_2002, zhu_automatic_2006, nguyen_ten_2019, greenacre_principal_2022}.

\subsection{Computation with the discrete second derivative rule} \label{sec:derivative}

The discrete second derivative is defined as
\begin{equation}\label{equation:discrete_deriv}
    \kappa_{k}(\sval{X}) := \eval{X}_{k-1} -2\eval{X}_{k} + \eval{X}_{k+1},
\end{equation}
and we define the \emph{derivative rule} --- an elbow rule based on the discrete second derivative ---  as 
\begin{equation}\label{equation:elbow_rule}
      r(\sval{X})
      := \argmax_{k \in \{1,\dots,p-2\}} \kappa_{k+1}(\sval{X}),
\end{equation}
which will be unique almost surely.
Thus, the $\left(r(\sval{X})+1\right)$th singular value has the largest discrete second derivative, and so the elbow occurs at the $\left(r(\sval{X})+1\right)$th index. This implies that the last $p-r(\sval{X})$ principal components capture predominantly noise, so only the first $r(\sval{X})$ principal components are retained for further study.

We now explicitly characterize the set \eqref{equation:elbow_solution_set} in the case of the derivative rule \eqref{equation:discrete_deriv}--\eqref{equation:elbow_rule}; it takes the form of one or two compact intervals.

\begin{proposition}\label{prop:elbow_solution_set}
    Let $X \in \RR^{n \times p}$, and $k \leq r(\sval{X})$, for $r(\cdot)$ defined in~\eqref{equation:elbow_rule}.  We define the values
    \begin{align*}
        c_{1,k}(X) &:= 
        \begin{cases} 
          \underset{i \in \cbr{2, \dots, k-2}}{\max} \kappa_i\rbr{\sval{X}} & \mathrm{if}\,k \geq 4 \\
          -\infty & \mathrm{otherwise}
          \end{cases},\\
        % \quad\mathrm{and}\quad
        c_{2,k}(X) &:= \begin{cases} 
          \underset{i \in \cbr{k+2, \dots, p-1}}{\max} \kappa_i\rbr{\sval{X}} & \mathrm{if}\,k \leq p-3 \\
          -\infty & \mathrm{otherwise}
          \end{cases},
    \end{align*}
    where $\kappa_i(\sval{X})$ is defined in~\eqref{equation:discrete_deriv}.
    Define the intervals
    \begin{align*}
        A_k(X) &:= [\max\{ {\textstyle \frac13} \rbr{\sval{X}_{k-1}^2 + 3\sval{X}_{k+1}^2 - \sval{X}_{k+2}^2}, \\
        &\quad\quad\,\, c_{1,k}(X) + 2\sval{X}_{k+1}^2 - \sval{X}_{k+2}^2\}, \infty) \quad\mathrm{for}\,k\geq 1,\\
        B_2(X) &:= \left({\textstyle \frac12}\rbr{\sval{X}_{1}^2 + \sval{X}_{3}^2 - c_{2,2}(X)}, \infty\right),\\
             B_k(X) &:= \sbr{{\textstyle \frac12}\rbr{\sval{X}_{k-1}^2 + \sval{X}_{k+1}^2 - c_{2,k}(X)}, 2\sval{X}_{k-1}^2 - \sval{X}_{k-2}^2 + c_{2,k}(X)} \quad\mathrm{for}\,k\geq 3.
    \end{align*}
    It holds that
    \begin{equation*}
        \cbr{t: k \leq r\left(\simpute{X}{t}^k\right)} =
        \begin{cases} 
          [0, \infty) & \mathrm{if}\, k=1\\
          A_k(X) & \mathrm{if}\, c_{1,k}(X) > c_{2,k}(X) \,\mathrm{or}\, |c_{2,k}(X)| = \infty \\
          B_k(X) & \mathrm{if}\, k\geq 3 \,\mathrm{and} \\
          &\sval{X}_{k-2}^2 - 2\sval{X}_{k-1}^2 > - 2 \sval{X}_{k+1}^2 + \sval{X}_{k+2}^2 \\
          A_k(X) \cup B_k(X) & \mathrm{otherwise}. \\
       \end{cases}
    \end{equation*}
    Thus, $\cbr{t: k \leq r\left(\simpute{X}{t}^k\right)}$ is either an interval or the union of two intervals.
\end{proposition}

The key insight of~\cref{prop:elbow_solution_set} is that to compute $\PP_{\delta, k}(X)$ in \eqref{equation:def-p}, we can  numerically integrate $h(t; \delta, \sval{X}_{[-k]})$ in \eqref{equation:h} over either one or two intervals. 
We note that in computing $\PP_{\delta, k}(X)$, the set $\cbr{t: k \leq r\left(\simpute{X}{t}^k\right)}$  will be intersected with $[s_{k+1}(X), s_{k-1}(X)]$; thus, values outside this interval are not of interest. \cref{prop:elbow_solution_set} is illustrated in~\cref{fig:elbow_rule_cases}.

\begin{figure}[!htb]
    \label{fig:elbow_rule_cases}
    \centering
    \includegraphics{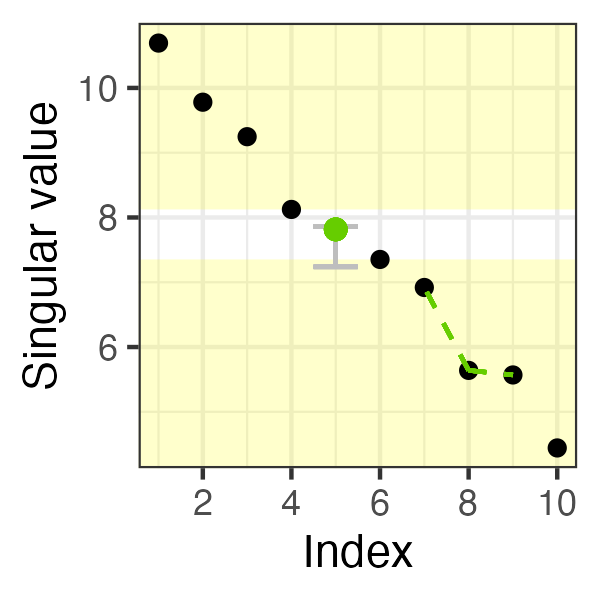}
    \includegraphics{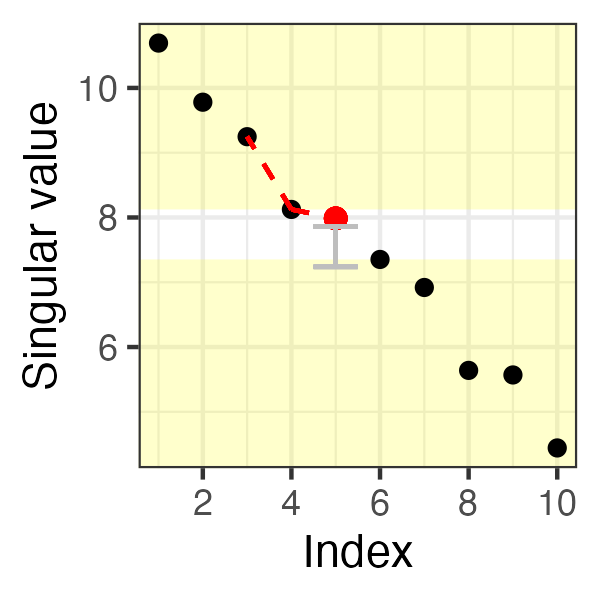}
    \includegraphics{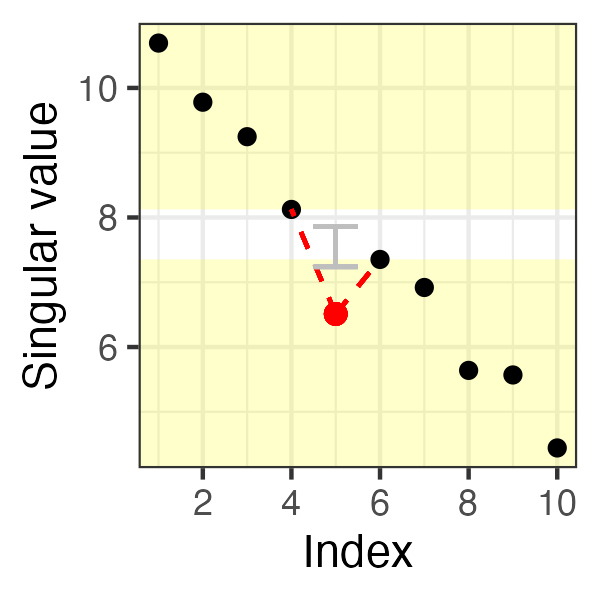}
    \caption{An illustration of~\cref{prop:elbow_solution_set} on a simulated matrix $x$ with $p=10$. 
    The derivative rule \eqref{equation:discrete_deriv}--\eqref{equation:elbow_rule} yielded $r(\sval{x}) = 7$, i.e., an elbow at $8$.
    We consider inference on the $5$th singularvalue. Each panel displays the following quantities: (i) in black, all singularvalues except the fifth, i.e., $\sval{x}_1,\ldots,\sval{x}_4, \sval{x}_6,\ldots,\sval{x}_{10}$; (ii) in white, the region in which  $\sval{x}_5$ must lie (since singularvalues are ordered, it must lie  between $\sval{x}_4$ and $\sval{x}_6$); (iii) as a gray bar, the interval $\cbr{t: 5 \leq r\left(\simpute{x}{t}^5\right)}$ from \cref{prop:elbow_solution_set}. 
    \emph{(Left):} The true value of $\sval{x}_5$ is shown as the green dot; note that it lies within the gray bar. For this value of $\sval{x}_5$, the elbow is at 8 (indicated with the dashed green lines). \emph{(Center):} A hypothetical value of $\sval{x}_5$ is shown as the red dot.  Now the elbow is at 4 (indicated with the dashed red lines). Thus, this hypothetical value of $\sval{x}_5$ falls outside of the gray bar depicting the interval $\cbr{t: 5 \leq r\left(\simpute{x}{t}^5\right)}$. \emph{(Right):} As in the center panel, but now the elbow is at 5, so $r\left(\simpute{x}{t}^5\right)=4$, i.e., this hypothetical value of $\sval{x}_5$ falls outside of the gray bar depicting the interval $\cbr{t: 5 \leq r\left(\simpute{x}{t}^5\right)}$. }
\end{figure}

\subsection{Computation with the  \citet{zhu_automatic_2006} elbow rule} \label{sec:zg}
Next, we consider an elbow rule motivated by the proposal of  
\citet{zhu_automatic_2006}; we call this the \emph{ZG rule}.
 \citet{zhu_automatic_2006}
 separately model the first $k$ and last $p-k$ squared singular values with Gaussian densities, and then select the index $k$ that maximizes the resulting density: that is, 
 \begin{equation}\label{equation:zg_rule}
    r(\sval{X})
      := 1 + \argmax_{k \in \{1,\dots,p-1\}} \ell_{k}(\sval{X})
\end{equation}  
where
\begin{align}
\label{equation:zg_rule_likelihood}
    \ell_k(\sval{X})
    &= \sum_{i=1}^k \log f_\Theta\rbr{\eval{X}_i; \hat\mu_{k,1}(\eval{X}), \hat\sigma^2_k(\eval{X})}\nonumber\\
    &+ \sum_{i=k+1}^p \log f_\Theta\rbr{\eval{X}_i; \hat\mu_{k,2}(\eval{X}), \hat\sigma^2_k(\eval{X})},
\end{align}
\normalsize
and 
$f(\cdot; \mu, \sigma^2)$ is a normal density with mean $\mu$ and variance $\sigma^2$. Here $\hat\mu_{k,1}$ and $\hat\mu_{k,2}$ are the sample means of the first $k$ and last $p-k$ squared singular values, respectively, and $\hat\sigma^2_k(\eval{X})$ is a pooled estimate of the variance.

Unlike the derivative rule, the ZG rule in~\eqref{equation:zg_rule} does not immediately yield a closed-form characterization of the  set in~\eqref{equation:elbow_solution_set}. Instead, we will numerically evaluate this set.
The likelihood in \eqref{equation:zg_rule_likelihood} is continuous, and for $k \in \cbr{2, \dots, p}$, the $k$th singular value is bounded. We can numerically approximate the set $\{ t: k \leq r\rbr{\simpute{X}{t}^k} \}$ by performing a line search over candidate values of $t$ between $\sval{X}_{k+1}$ and $\sval{X}_{k-1}$, and manually checking whether the condition $k \leq r\rbr{\simpute{X}{t}^k}$ is satisfied. We then compute $\PP_{\delta, k}(X)$ via numerical integration over these regions. In practice, \eqref{equation:elbow_solution_set} appears to be the union of one or two intervals.

\section{Empirical results}\label{sec:simulation_study}

\verb=R=-language scripts to reproduce all numerical results in this section are available at \newline\if0\blind
\href{https://github.com/rflperry/elbow_inference}{https://github.com/rflperry/elbow\_inference}.
\fi
\if1\blind
\href{https://github.com/XXXX}{https://github.com/XXXX}.
\fi

In Section~\ref{subsec:sim-datagen}, we describe our data generation procedure. In~\cref{subsec:naive_vs_estimand}, we compare our population PVE in~\cref{eq:pop-pve} to the sample PVE in~\cref{eq:pve}. In Sections~\ref{subsec:sim-hypothtest} and~\ref{subsec:sim-confint}, respectively, we investigate the selective type 1 error control and power of the hypothesis tests developed in Section~\ref{sec:hypothesis_testing}, and the selective coverage and width of the confidence intervals developed in Section~\ref{sec:confidence_interval}.
We present an application to gene expression data in Section~\ref{sec:application}. 
Throughout this section, we conduct inference only on indices selected by the elbow rule, i.e., those for which $k \leq r(\sval{X})$. 
In what follows, we use the ZG variant of the elbow rule, defined in Section~\ref{sec:zg}. 
See~\cref{sec:appendix_deriv_rule} for  experiments with the derivative rule, defined in Section~\ref{sec:derivative}.

\subsection{Data generation}
\label{subsec:sim-datagen}
We simulate data according to the model specified in~\eqref{equation:mvn_model}, with $n=50$ and $p=10$. We sample mean matrices $\Theta = U_\Theta S_\Theta V_\Theta^\top$ by setting $S_\Theta$ to be a diagonal matrix, and sampling singular vectors $U_\Theta$ and $V_\Theta$ from the singular vectors of a random matrix of i.i.d. $\Norm(0, 1)$ entries. Under the \textit{global null}, we set $S_\Theta$, and hence $\Theta$, to be $0_{n \times p}$. Under the \textit{alternative} (further details in \cref{sec:appendix_zg_rule}), we set all but the first five  diagonal elements of $S_\Theta$ to equal zero, so that $\rank(\Theta)=5$.
The value of $\sigma$ in \eqref{equation:mvn_model} varies across simulations.

\subsection{A comparison of $\PVEhat{X}_k$ and $\PVE{X}_k$}\label{subsec:naive_vs_estimand}

We first examine how the amount of noise affects the discrepancy between the naive quantity $\PVEhat{X}_k$ in~\cref{eq:pve} and the population quantity $\PVE{X}_k$ in~\cref{eq:pop-pve}. Recall that when $\sigma = 0$, and thus $X = \Theta$, these two quantities are identical for all $k$. As $\sigma$ increases, and for each index $k \leq r(\sval{X})$, \cref{fig:naive_vs_pve} displays the log of the median of $\PVEhat{X}_k / \PVE{X}_k$ across simulated datasets. We make two observations. First, for $k \leq 5$, as $\sigma$ increases, the log ratio increases, since $\PVEhat{X}_k$ increasingly overestimates the population quantity $\PVE{X}_k$. Second, when $\sigma$ is small, the ZG rule never yields $r(\sval{X}) > 5$, and so no ratios are displayed when $\sigma$ is small and $k > 5$. 

\begin{figure}[!htb]
    \centering
    \includegraphics[width=\linewidth]{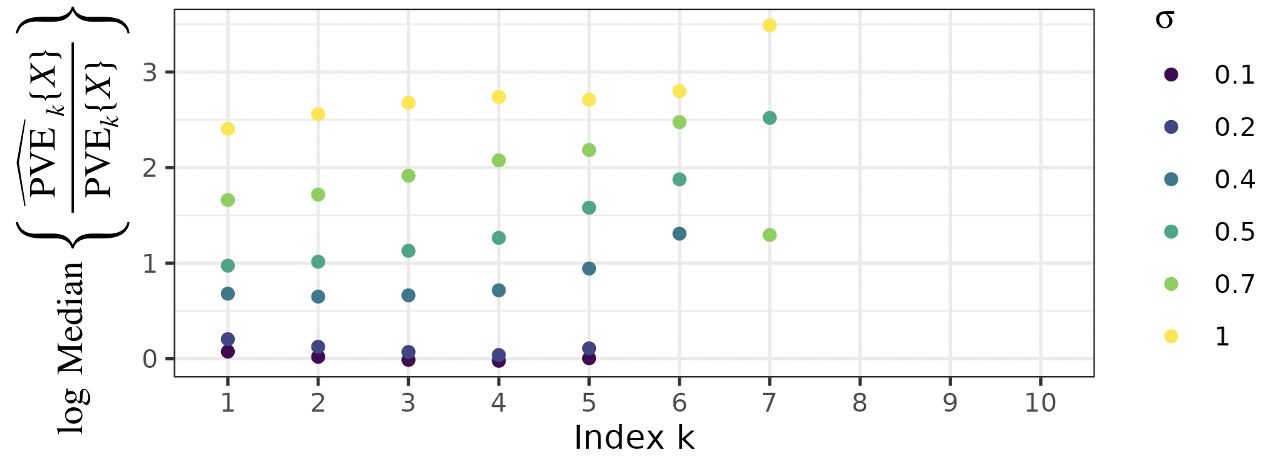}
    \caption{
    We plot the empirical median, across $10^3$ simulated datasets, of the ratio of $\PVEhat{X}_k$ to $\PVE{X}_k$ across values of $k$ and variances $\sigma$, where the ratio is only compared when $k \leq r(\sval{X})$ under the ZG rule.
    At $\sigma = 0$, these quantities are identical; thus, at low values of $\sigma$ the log ratio is close to zero. As $\sigma$ increases, the matrix $X$ becomes increasingly noisy and $\PVEhat{X}_k$ becomes larger than $\PVE{X}_k$; this is in agreement with the fact that $\sum_{k=1}^p\PVE{X}_k \leq 1$, attaining $1$ if the estimated singular vectors are the population singular vectors, i.e., if $\sigma = 0$.
    }
    \label{fig:naive_vs_pve}
\end{figure}

\subsection{Hypothesis testing simulations}
\label{subsec:sim-hypothtest}

\subsubsection{Selective type 1 error control under the global null}
We now investigate whether the test of $H_{0, k}(X): \PVE{X}_k=0$ developed in Section~\ref{sec:hypothesis_testing} controls the selective type 1 error, defined as \eqref{eq:selT1E}.  We generate data under the global null hypothesis (i.e., $\Theta = 0_{n\times p}$, so that $\PVE{X}_k=0$ for all $X$ and all $k$, by definition) and $\sigma^2=1$, and apply the ZG rule to select an index $r(\sval{X})$, and then test $H_{0, k}(X)$  at level $\alpha=0.1$ for $k \leq r(\sval{X})$.

The  qq-plot in the right-hand panel of~\cref{fig:toy_example} displays the distribution of p-values under the global null hypothesis. The plot is stratified by the value of $k$, corresponding to the null hypothesis $H_{0,k}(X): \PVE{X}_k=0$ tested. 
The p-values that account for selection using the ZG rule follow the $45^{\circ}$ line, and thus are uniformly distributed under the null. Furthermore, the p-values that do not account for selection using the ZG rule follow the $45^\circ$ line when $k=1$, i.e., when the first principal component is tested: this is because the ZG rule \emph{always} selects the first principal component, and so no adjustment for the ZG rule is required. However, as $k$ increases, the latter p-values become increasingly anticonservative, due to failure to account for the fact that we test $H_{0,k}(X)$ only for $k \leq r(\sval{X})$. In other words, the test that does not account for the ZG rule has poorest error control when the ZG rule incorrectly suggests that there is substantial signal in the matrix, i.e., in the exact setting where  error control is important. The right hand panel of~\cref{fig:toy_example} clearly indicates that failure to account for the ZG rule leads to anti-conservative p-values, and therefore the corresponding tests have inflated type 1 error. Thus, in the rest of this section, we only show results for tests, confidence intervals, and point estimates that \emph{do} account for selection via the ZG rule.

\subsubsection{Selective power under the alternative}

To evaluate the selective power of the test of $H_{0, k}(X): \PVE{X}_k=0$ under the alternative in~\cref{subsec:sim-datagen}, we simulate datasets for $\sigma \in \{0.1, 0.2, 0.4, 0.5, 0.7, 1\}$ (recall that $\rank(\Theta)=5$, with details in \cref{sec:appendix_zg_rule}). We define the following metrics:
\begin{align*}
    \text{Detection probability }
    &= \frac{\text{\# [datasets where }r(\sval{X}) = \rank(\Theta) ] }{\text{\# datasets}},\\
    \text{Selective power at $k$} &= \frac{~\text{\# [datasets where we reject }H_{0,k}(X)\text{ and } k \leq r(\sval{X})]} {\text{\# [datasets where }k \leq r(\sval{X})]}.
\end{align*}
The detection probability is the fraction of simulated datasets for which the ZG rule selects the correct rank, whereas the selective power at $k$ is the fraction of simulated datasets for which we reject $H_{0,k}(X)$, out of those for which the ZG rule selects at least rank $k$. 

In the left-hand  panel of~\cref{fig:zg-power}, we plot the detection probability of the ZG rule against the inverse variance. As the variance decreases, the probability that the ZG rule selects the correct rank increases. The right-hand panel displays the selective power of the selective test as a function of the inverse variance, stratified by the value of $k$ in the tested null hypothesis $H_{0,k}(X): \PVE{X}_k = 0$.
We see that for a given value of $k$, the power of each test increases as the variance decreases. Furthermore, as expected, the power to reject  $H_{0,k}(X): \PVE{X}_k = 0$ decreases as the index $k$ increases (since each subsequent principal component captures a decreasing proportion of the variance in the data). Note that for all $k$, unless $\Theta$ is the zero matrix, since the $k$th estimated singular vectors are almost surely not the population singular vectors, they will still contain some signal, i.e. $\pivot{X}_k \neq 0$; thus even though $\rank(\Theta) = 5$, we have power at $k=6$.

\begin{figure}[!htb]
    \centering
    \includegraphics[]{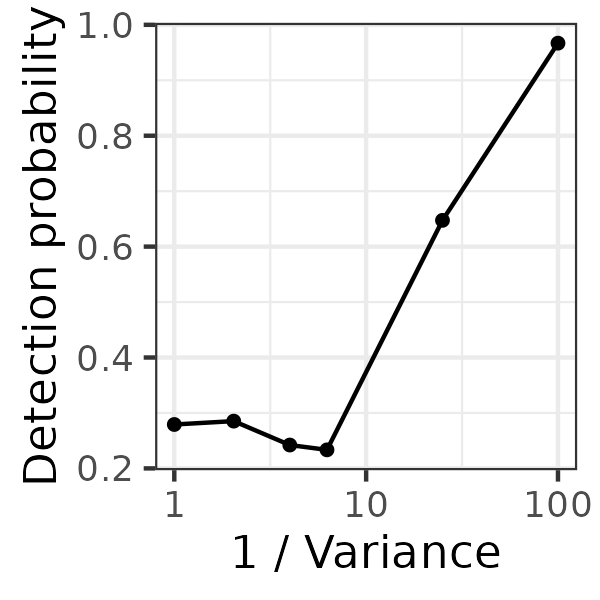}
    \includegraphics[]{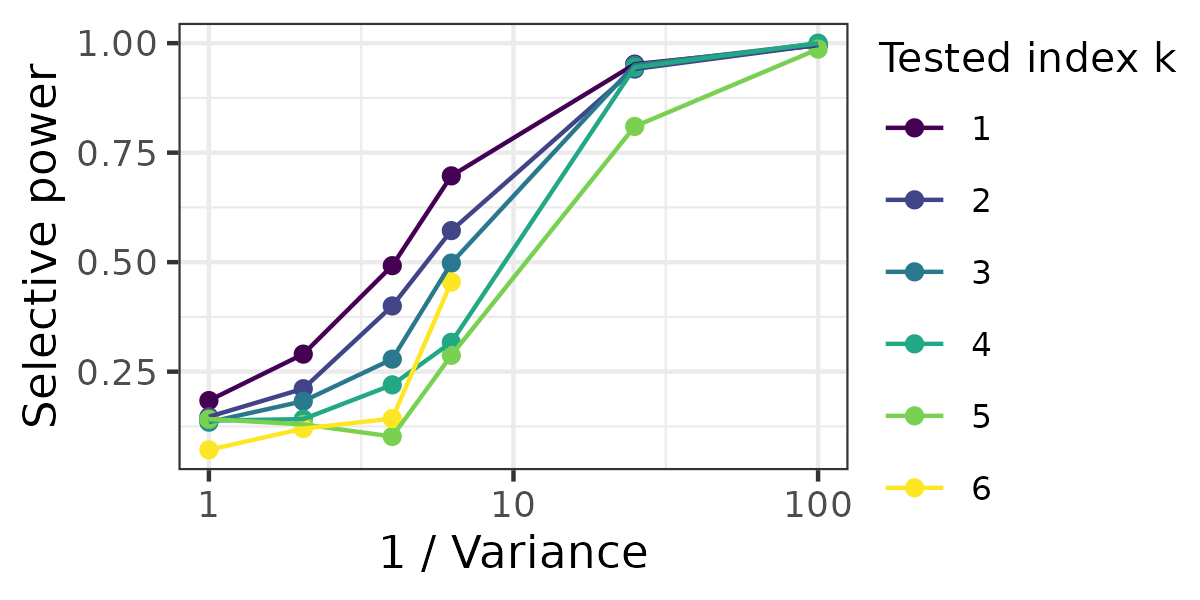}
    \caption{\emph{(Left):} The proportion of simulated datasets for which the ZG rule selects the true rank of the matrix increases with decreasing variance $\sigma^2$. \emph{(Right):} The selective power for each tested index $k  \leq r(\sval{x})$ increases with decreasing variance $\sigma^2$; furthermore, there is higher power at smaller tested indices $k$. Note that index $6$ is never tested below a certain variance since the ZG rule does not select it.}
    \label{fig:zg-power}
\end{figure}

\subsection{Confidence interval simulations}
\label{subsec:sim-confint}

We now investigate the two-sided confidence interval~\cref{equation:si_confidence_interval} developed in Section \ref{sec:confidence_interval} for $\PVE{X^{(1)}}_k$, defined in~\cref{eq:pop-pve}, where $k$ is selected by the ZG rule. For a fixed $\alpha$, we choose $\alpha_1 = (3/4)\alpha$ and $\alpha_2 = (1/4)\alpha$, which empirically yielded narrow intervals on average; this narrows the confidence interval of the numerator~\cref{eq:CI-numer} at the cost of widening the confidence interval of the denominator~\cref{eq:CI-denom}. We compute confidence intervals \emph{only} for $k \leq r(\sval{X^{(1)}})$, i.e., for indices selected by the ZG rule. 
Thus, we always conduct inference on $\PVE{X^{(1)}}_1$  (since the ZG rule always selects $r(\sval{X^{(1)}}) \geq 1$; see \eqref{equation:zg_rule}), while inference on $\PVE{X^{(1)}}_k$ for $k>1$ occurs only in a subset of  simulation replicates, and the cardinality of this subset decreases in $k$. 

To assess coverage, we simulate datasets under the alternative in~\cref{subsec:sim-datagen} with $\sigma=0.1$, and vary $\alpha \in \cbr{0.1, 0.3, 0.5, 0.7, 0.9}$.
The selective coverage at index $k$, defined as 
\begin{equation}\label{eq:selective_coverage}
    \frac{\text{\# [datasets $X$ where }k \leq r(\sval{X^{(1)}}) 
    \text{ and } \PVE{X^{(1)}}_k \in \mathrm{CI}_{\alpha_1, \alpha_2,k}(X^{(1)}, X^{(2)})  ] }{{\text{\# [datasets $X$ where } k \leq r(\sval{X^{(1)}})} \text{]}},
\end{equation}
is displayed in~\cref{fig:zg-ci_coverage}. The results reflect that the coverage guarantee from~\cref{thm:si_ci_coverage} bounds the confidence interval coverage above $1-\alpha$; our coverage is conservative.

\begin{figure}[!htb]
    \centering
    \includegraphics[width=\linewidth]{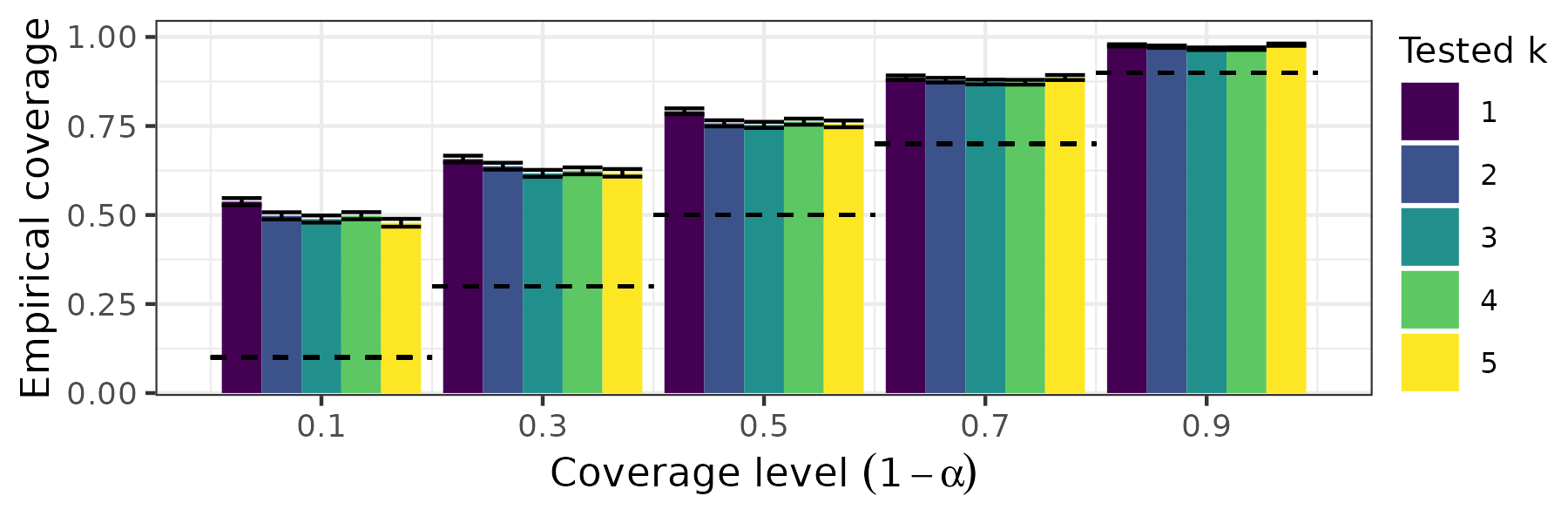}
    \caption{ We simulate $10^4$ datasets under the alternative in~\cref{subsec:sim-datagen}, with $\sigma = 0.1$. For each dataset $X$, we apply thinning and the ZG rule to obtain $r(\sval{X^{(1)}})$, and construct $1-\alpha$ confidence intervals for $\PVE{X^{(1)}}_k$ across a range of values of $\alpha$ and for $k \leq r(\sval{X})$. We display the proportion of simulated datasets for which the confidence interval contains the true value of $\PVE{X^{(1)}}_k$ (with bars depicting twice the standard error of this proportion). The confidence intervals have coverage above the nominal level (marked by horizontal dashed lines).} 
    \label{fig:zg-ci_coverage}
\end{figure}

Next, we investigate the endpoints and widths of these confidence intervals under the alternative in~\cref{subsec:sim-datagen}. We simulate $10^3$ datasets with each of two variance values $\sigma \in \{0.1, 0.2\}$, and compute confidence intervals accounting for selection by the ZG rule $k \leq r(\sval{X^{(1)}})$. In~\cref{fig:zg-ci_coverage_screeplot}, for each value of $k$ selected by the ZG rule, we plot the median upper and lower confidence interval bounds, the median value of $\PVE{X^{(1)}}_k$, and the median point estimate $\PVE{X^{(1)}, X^{(2)}}_{k, \mathrm{MLE}}$, as defined in (3.11). When the noise is low, the confidence intervals are relatively narrow and exclude $0$. When the noise level is higher, the intervals are wider, especially for larger values of $k$. In fact, when $\sigma=0.2$ and $k=5$, the confidence interval spans the entire range from $0$ to $1$, and so we  cannot determine whether the fifth estimated singular vector captures signal in $\Theta$.

Figure~\ref{fig:zg-ci_coverage_screeplot} and the left-hand panel of Figure~\ref{fig:toy_example} are closely related:  the former represents the average over $10^3$ simulated datasets, whereas the latter involves only one. 

\begin{figure}[!htb]
    \centering
    \includegraphics[width=\textwidth]{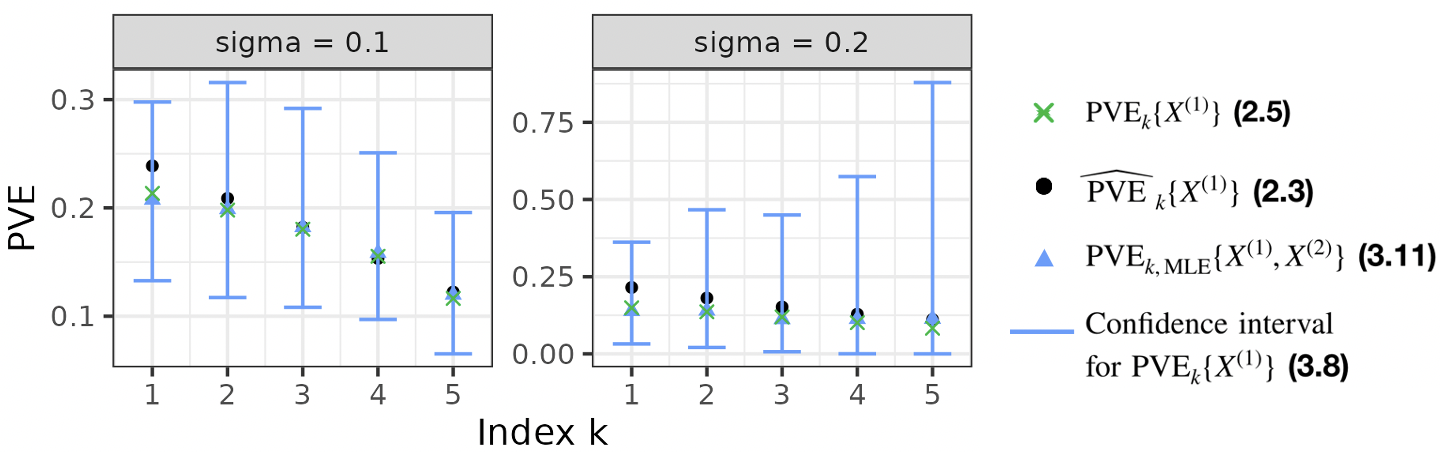}
    \caption{  We simulate $10^3$ datasets under low noise ($\sigma=0.1$, left) and high noise ($\sigma=0.2$, right), under the alternative in~\cref{subsec:sim-datagen} where $\text{rank}(\Theta)=5$. For each value of $k$, we conduct inference accounting for the ZG rule selection event $\cbr{k \leq r(\sval{X^{(1)}})}$.
    The scree plots display the median $\PVEhat{X^{(1)}}_k$ \eqref{eq:pve} (black dots), as well as the median population $\PVE{X^{(1)}}_k$ \eqref{eq:pop-pve} (green triangles). The medians of the confidence interval endpoints, as well as the median point estimate $\PVE{X^{(1)}, X^{(2)}}_{k, \mathrm{MLE}}$ \eqref{eq:mle_pve}, are also displayed.}
    \label{fig:zg-ci_coverage_screeplot}
\end{figure}

\begin{remark}
    In~\cref{fig:toy_example}, the matrix referred to as $X$ is in fact $X^{(1)}$ in the notation of~\cref{sec:confidence_interval} and beyond. That is, the confidence intervals reported are as defined in~\cref{equation:si_confidence_interval}.
\end{remark}

\subsection{An application to gene expression data}\label{sec:application}

We consider the Nutrimouse dataset~\citep{martin_novel_2007}, which consists of expression measurements of 120 genes across 40 mice. For convenience, we subset to the first $20$ genes, and center the columns (genes); by~\cref{prop:center}, our inference is thus on the centered mean matrix of the data. Since $\sigma^2$ is unknown, as in~\citet{choi_selecting_2017} we use the robust estimator proposed in~\cite{matan_optimal_2014},  defined as $\hat{\sigma}^2_{\text{med}}(X) := (\mathrm{median}\cbr{\sval{X}_1, \dots, \sval{X}_p})^2 / (n \mu_{n/p})$ where $\mu_{n/p}$ is the median of a Marchenko-Pastur distribution with parameter $n/p$. Under~\cref{equation:mvn_model}, and assuming that $\Theta$ is low-rank and further technical assumptions, this estimator is consistent for $\sigma^2$ as $p$ increases~\citep{matan_optimal_2014}. \cref{sec:appendix_zg_rule} provides empirical evidence in simulations that inference using this estimator is valid, but has lower power than when using the true value of $\sigma$.

We apply PCA to the thinned data matrix $X^{(1)}$. The ZG rule (Section~\ref{sec:zg}) selects $r(\sval{X^{(1)}})=3$; this agrees visually with the scree plot displayed in the left-hand panel of~\cref{fig:nutrimouse}. For each $k \in \cbr{1, 2, 3}$, at $\alpha_1 = 0.075$ and $\alpha_2 = 0.025$, the right-hand panel of~\cref{fig:nutrimouse} displays selective $90\%$ confidence intervals for $\PVE{X^{(1)}}_k$, defined in~\cref{eq:pop-pve}, along with the sample quantity $\PVEhat{X^{(1)}}_k$ from~\cref{eq:pve} and the estimate $\PVE{X^{(1)}, X^{(2)}}_{k, \mathrm{MLE}}$ from~\cref{eq:mle_pve}. None of the three confidence intervals include zero, thus rejecting the null hypothesis.

\begin{figure}[!htb]
    \centering
    \includegraphics[width=\linewidth]{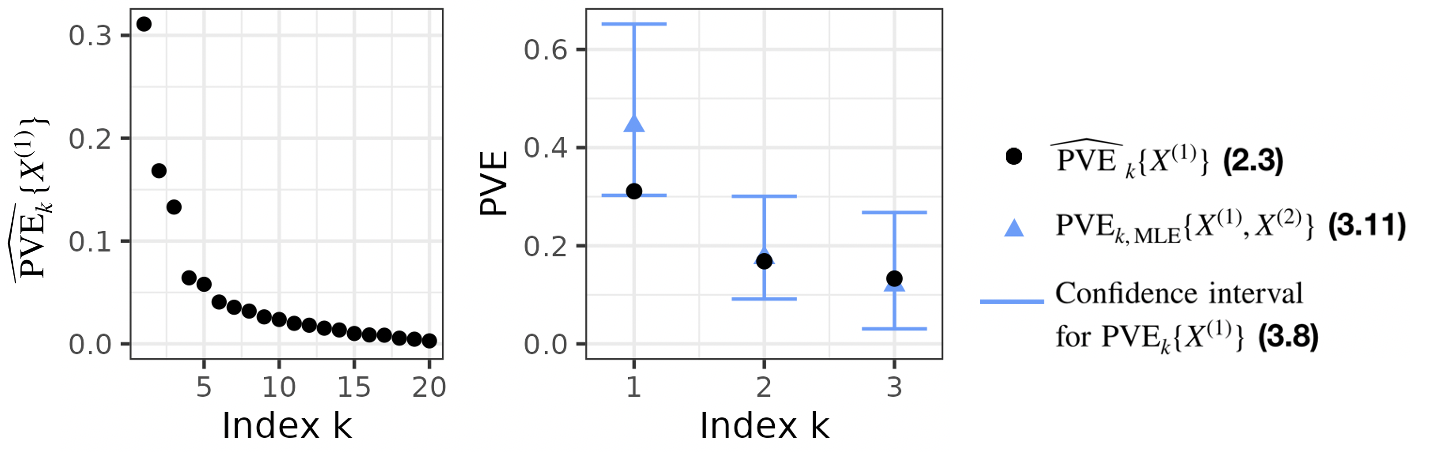} 
    \caption{ \emph{(Left):} The scree plot of the centered gene expression values for the Nutrimouse data. The ZG rule selects $r(\sval{x^{(1)}})=3$,  which corresponds to an obvious elbow. \emph{(Right):} For each $k \in \cbr{1, 2, 3}$, we plot selective $90\%$ confidence intervals for $\PVE{X^{(1)}}_k$, defined in~\cref{eq:pop-pve}, at $\alpha_1 = 0.075$ and $\alpha_2 = 0.025$, along with the sample quantity $\PVEhat{X^{(1)}}_k$ from~\cref{eq:pve} and estimates $\PVE{X^{(1)}, X^{(2)}}_{k, \mathrm{MLE}}$ from~\cref{eq:mle_pve}. All three confidence intervals exclude zero.
    }
    \label{fig:nutrimouse}
\end{figure}

\section{Discussion} \label{sec:discussion}

In this paper, we consider  inference on the PVE of the principal components of a dataset. While the sample PVE is often reported in scientific papers, we define a population parameter corresponding to this quantity, and propose an associated inferential framework. 

The population version of the PVE that we define reflects the PVE of the \emph{unknown population mean} --- not the observed data --- captured by the principal components. We further provide (i) tests of hypotheses involving the population PVE with valid type 1 error control, (ii) confidence intervals for the population PVE that  attain the nominal coverage, and (iii) point estimates for the population PVE.

In this paper, we conduct inference conditional on several events. The motivations for each of these conditioning events are distinct:
\begin{enumerate}
\item The population PVE that we propose is an unusual quantity, as it involves the $k$th singular vectors of $X$. Thus, to make this parameter deterministic we condition on these singular vectors. 
\item As in~\citet{choi_selecting_2017}, we further condition on all remaining singular vectors, as well as all but the $k$th singular value, to eliminate nuisance parameters. However, our final inferential guarantees hold even without conditioning on these quantities.
\item Data analysts typically are interested in  the PVE for only a subset of  principal components; this subset is often selected via the ZG rule, or via another function of the singular values. Thus, we further condition on the event that a given procedure, such as the ZG rule, selected this particular principal component for inference. 
\end{enumerate}

Multiple avenues of future work remain. First, we have assumed knowledge of the variance, $\sigma^2$, in \eqref{equation:mvn_model}. Of course,  this quantity  is typically unknown and must be estimated. Possible estimators include the scaled squared median of the singular values~\citep{matan_optimal_2014} or the minimum mean of squared residuals identified through cross-validation~\citep{choi_selecting_2017}. Moreover, the power of the proposed approach could be increased by introducing auxiliary randomization, as in \cite{panigrahi_exact_2024}. Lastly, it is an open question how to best choose the data thinning constant $c$, as well as $\alpha_1$ and $\alpha_2$.

\paragraph{Supplementary Materials}
The online supplement contains the proofs of results in the main paper, as well as expanded simulation results.

\paragraph{Disclosure Statement}
The authors report there are no competing interests to declare.

\bibliographystyle{plainnat}
\bibliography{references}

\newpage

\appendix

\section{Theory and proofs}\label{appendix:sec:proofs}

\subsection{ Inference on a centered matrix}\label{appendix:subsec:centering}

In this paper, we have assumed that $X \sim \text{MN}(\Theta, \sigma^2 I_n \otimes I_p)$ (see \eqref{equation:mvn_model}). However, in practice, PCA is typically performed on a column-centered matrix, $CX$, where
 $C := I_n - \frac{1}{n}1_n1_n^\top$ is a rank-$(n-1)$ projection matrix.
Thus, % Since $C$ is idempotent, % we see that 
\begin{equation*}
    CX \sim \text{MN}(C \Theta, C,  \sigma^2 I_p).
\end{equation*}
 The results in this paper cannot be directly applied to $CX$, as its distribution does not take the form  in \eqref{equation:mvn_model}. 
To be explicit, let $\mathrm{PVE}_k^{\Theta} \{X\}$ denote the PVE of the $k$th singular vectors of $X$ with respect to $\Theta$. The following result shows that we can instead apply the results in this paper to $H^\top X$, where $H$ is a matrix square root of $C$, to conduct inference on $\mathrm{PVE}_k^{C \Theta} \{C X\}$. 

\begin{proposition}[Inference after column-centering $X$]\label{prop:center}
    Let $X \sim \text{MN}(\Theta, I_n, \sigma^2 I_p)$ and define the centering matrix $C := I_n - \frac{1}{n}1_n1_n^\top$. 
    Furthermore, take $H$ to be a $n \times (n-1)$ matrix square root of $C$ that satisfies $H H^\top = C$ and $H^\top H= I_{n-1}$.  Then, (i)
    $H^\top X  \sim \text{MN}(H^\top \Theta, I_{n-1}  \otimes \sigma^2 I_p)$; 
and (ii) $\mathrm{PVE}_k^{H^\top \Theta} \{H^\top X\} = \mathrm{PVE}_k^{C \Theta} \{C X\}$.
 \end{proposition}

%Proposition~\ref{prop:center}   implies that if our interest lies in the parameter 
%$\mathrm{PVE}_k^{C\Theta} \{CX\} $ 
%associated with the column-centered version of $X$, then we can apply the inference tools developed in this paper to $H^\top X$ (which follows the correct distribution) in order to conduct inference on $\mathrm{PVE}_k^{H^\top \Theta} \{H^\top X\}$, as these two quantities are equal. 

\begin{proof}
Claim (i) follows directly from the fact that $X \sim \text{MN}(\Theta, I_n, \sigma^2 I_p)$ and that $H^\top H= I_{n-1}$. 

We will now establish claim (ii). 
Recall that $\Usvd{A} \diag(\sval{A}) \Vsvd{A}^\top$ denotes the SVD of the matrix $A$. 
Thus, 
\begin{eqnarray} CX = H (H^\top X) &=& H \left( \Usvd{H^\top X} \diag(\sval{H^\top X}) \Vsvd{H^\top X }^\top \right) \nonumber \\
&=& \left( H \Usvd{H^\top X} \right) \diag(\sval{H^\top X}) \Vsvd{H^\top X }^\top.
\end{eqnarray}
Furthermore, $\left(H \Usvd{H^\top X}) \right)^\top H \Usvd{H^\top X} = I_p$. Thus, it follows that  the SVD of $CX$ is closely related to the SVD of $H^\top X$: that is,  
\begin{equation}\Usvd{CX} =H \Usvd{H^\top X},\quad \Vsvd{CX} = \Vsvd{H^\top X}, \quad\sval{CX}=\sval{H^\top X}. \label{eq:svd-equal} \end{equation}

Thus, the singular values of $CX$ and $H^\top X$ are equal. An important corollary of this is that 
\begin{equation*}
    \nbr{ C \Theta}_F^2 = \sum_{k=1}^p \sval{C \Theta}_k^2 = \sum_{k=1}^p \sval{H^\top \Theta}_k^2 = \nbr{ H^\top \Theta}_F^2.
\end{equation*}
Furthermore,  
\begin{eqnarray*}
\Usvd{CX}_k^\top \left( C \Theta \right) \Vsvd{CX}_k 
&=&  \Usvd{H^\top X}_k^\top H^\top \left(  C \Theta \right) \Vsvd{H^\top X}_k\\
&=& \Usvd{H^\top X}_k^\top \left( H^\top \Theta \right) \Vsvd{H^\top X}_k.
\end{eqnarray*}
Thus, the quantities \begin{equation}
\mathrm{PVE}_k^{C\Theta} \{CX\} =    \frac{ \left( \Usvd{CX}_k^\top C\Theta \Vsvd{CX}_k \right)^2 }{ \nbr{C \Theta}_F^2 }
\end{equation}
and \begin{equation*}
    \mathrm{PVE}_k^{H^\top \Theta} \{H^\top X\} =    \frac{ \left( \Usvd{H^\top X}_k^\top H^\top \Theta \Vsvd{H^\top X}_k \right)^2 }{ \nbr{H^\top \Theta}_F^2 }
\end{equation*}
are equal. 
\end{proof}

\subsection{Proofs of Section 3 results}

\begin{proof}[Proof of~\cref{lemma:si_density}]
    We follow the arguments laid out in~\citet{choi_selecting_2017}, starting from the joint density of the SVD of $X$ in~\eqref{equation:joint_density}. We expand and collect conditioned-upon terms into the implicit constant of proportionality as follows:
    \begin{equation}
        \label{equation:sk_density}
        \begin{split}
            &f_\Theta\rbr{\sval{X}_{k}~|~\sval{X}_{[-k]}, \Usvd{X}, \Vsvd{X}}\\
                &\propto \exp\bigg\{
                    -\frac{1}{2\sigma^2} \Bigg[ \rbr{\sval{X}_k^2 + \sum_{k': k'\neq k} \sval{X}_{k'}}\\
                    &\qquad\quad- 2 \rbr{\sval{X}_k \pivot{X}_k + \sum_{k' \in [p]: k' \neq k} 
                    % \sval{x}_k \Usvd{x}_k^T \Theta\Vsvd{x}_k
                    \sval{X}_{k'} \pivot{X}_{k'}}
                    \Bigg]
                + \tr(\Theta^\top \Theta) \\
                &\qquad\quad\bigg\}\times J\rbr{\sval{X}}\, d\mu_{n,p}\rbr{\Usvd{X}} \, d\mu_{p,p}\rbr{\Vsvd{X}}
                \II_{[\sval{X}_{k+1}, \sval{X}_{k-1}]}(\sval{X}_k)\\
            &\propto \exp\rbr{
                    -\frac{1}{2\sigma^2} \sbr{\sval{X}_k^2 - 2 \sval{X}_k \pivot{X}_k 
                    % \sval{x}_k \Usvd{x}_k^T \Theta\Vsvd{x}_k
                    }
                }\,
                \times J\rbr{\sval{X}} \times
                \II_{[\sval{X}_{k+1}, \sval{X}_{k-1}]}(\sval{X}_k)\\
            &\propto \exp\rbr{
                    -\frac{1}{2\sigma^2} \sbr{\sval{X}_k^2 - 2 \sval{X}_k \pivot{X}_k 
                    % \sval{x}_k \Usvd{x}_k^T \Theta\Vsvd{x}_k
                    }
                }\\
                &\qquad\times
                \sbr{\sval{X}_k^{n-p} \times \Pi_{k' \in [p]: k' \neq k} \abr{\sval{X}_k^2 - \sval{X}_{k'}^2}} \times
                \II_{[\sval{X}_{k+1}, \sval{X}_{k-1}]}(\sval{X}_k),
        \end{split}
    \end{equation}
    where we have made use of the definition of $J(\cdot)$ from~\eqref{equation:J(s)}.
    Applying Bayes' rule, we see that
    \begin{equation}
        \begin{split}
            &f_\Theta\rbr{\sval{X}_{k}~|~\sval{X}_{[-k]}, \Usvd{X}, \Vsvd{X}, k \leq r\rbr{\sval{X}}}\\
            %&\propto
            %    f_\Theta\rbr{\sval{X}_{k}~|~\sval{X}_{[-k]}, \Usvd{X}, \Vsvd{X}}
            %    \times 
            %    f_\Theta\rbr{k \leq r\rbr{\sval{X}}~|~\sval{X}_{[-k]}, \Usvd{X}, \Vsvd{X}}\\
            &\propto f_\Theta\rbr{\sval{X}_{k}~|~\sval{X}_{[-k]}, \Usvd{X}, \Vsvd{X}} \times \II[k \leq r\rbr{\sval{X}}].
        \end{split}
    \end{equation}
    This is the density in~\eqref{equation:sk_density} truncated to the elbow rule selection event.
\end{proof}

\begin{proof}[Proof of~\cref{prop:theta_delta_k}]
    As defined in~\eqref{equation:h},
    \begin{align*}
        h\rbr{t; \delta, \sval{x}_{[-k]}} := \exp{\left(-\frac{t^2}{2\sigma^2} + \frac{t}{\sigma^2}\delta\right)} t^{n-p} \prod_{k' \in [p]: k' \neq k} \left|t^2 - \sval{x}_{k'}^2\right|.
    \end{align*}
    Furthermore, recall from~\cref{equation:s_impute} that
    \begin{equation*}
        \simpute{x}{t}^k := (\sval{x}_1, \dots, \sval{x}_{k-1}, t, \sval{x}_{k+1}, \dots, \sval{x}_p)
    \end{equation*}
    Note that
    \begin{equation*}
        h\rbr{\sval{x}_k; \pivot{x}_k, \sval{x}_{[-k]}}  \II_{[\sval{x}_{k+1}, \sval{x}_{k-1}]}(\sval{x}_k) \II\left[k \leq r\left(\simpute{x}{\sval{x}_k}^k\right) \right]
    \end{equation*}
    is equal to the right-hand side of~\eqref{equation:si_density} and thus is proportional to the conditional density
    \begin{equation*}
        f \left(\sval{x}_k~|~\Usvd{X}=\Usvd{x}, \Vsvd{X}=\Vsvd{x}, \sval{X}_{[-k]}=\sval{x}_{[-k]}, k \leq r(\sval{X}) \right).
    \end{equation*}
    Therefore, for any $x \in \RR^{n \times p}$,
    \begin{align}
        \label{equation:exact_density}
        &f \left(\sval{x}_k~|~\Usvd{X}=\Usvd{x}, \Vsvd{X}=\Vsvd{x}, \sval{X}_{[-k]}=\sval{x}_{[-k]}, k \leq r(\sval{X}) \right)\\ \nonumber
        &= \frac{
            h\rbr{\sval{x}_k; \pivot{x}_k, \sval{x}_{[-k]}} \II_{[\sval{x}_{k+1}, \sval{x}_{k-1}]}(\sval{x}_k) \II\left[k \leq r\left(\simpute{x}{\sval{x}_k}^k\right) \right]
        }{
            \int_{-\infty}^{\infty} h\rbr{t; \pivot{x}_k, \sval{x}_{[-k]}} \II_{[\sval{x}_{k+1}, \sval{x}_{k-1}]}(t) \II\left[k \leq r\left(\simpute{x}{t}^k\right) \right] dt
        },
    \end{align}
    where the denominator is the normalization constant required for the density to integrate to one. It follows that
    \begin{align*}
    &\prob \big(
        \sval{X}_k \geq \sval{x}_k~\big|~ \Usvd{X}=\Usvd{x}, \Vsvd{X}=\Vsvd{x},
        \sval{X}_{[-k]}=\sval{x}_{[-k]},
         k \leq r(\sval{x})
         \big)\\
        &= \int_{\sval{x}_k}^\infty f \left(t~|~\Usvd{x}, \Vsvd{x}, \sval{x}_{[-k]}, k \leq r(\sval{x}) \right)dt\\
        &= \int_{\sval{x}_k}^\infty  \rbr{\frac{
            h\rbr{t; \pivot{x}_k, \sval{x}_{[-k]}} \II_{[\sval{x}_{k+1}, \sval{x}_{k-1}]}(t) \II\left[k \leq r\left(\simpute{x}{\sval{x}_k}^k\right) \right]
        }{
            \int_{-\infty}^{\infty} h\rbr{t; \pivot{x}_k, \sval{x}_{[-k]}} \II_{[\sval{x}_{k+1}, \sval{x}_{k-1}]}(t) \II\left[k \leq r\left(\simpute{x}{\sval{x}_k}^k\right) \right] dt
        }} dt\\
        &= \frac{
            \int_{\sval{x}_k}^{\sval{x}_{k-1}} h\rbr{t; \pivot{x}_k, \sval{x}_{[-k]}} \II\left[k \leq r\left(\simpute{x}{\sval{x}_k}^k\right) \right] dt
        }{
            \int_{\sval{x}_{k+1}}^{\sval{x}_{k-1}} h\rbr{t; \pivot{x}_k, \sval{x}_{[-k]}} \II\left[k \leq r\left(\simpute{x}{\sval{x}_k}^k\right) \right] dt
        },
    \end{align*}
    where the second equality follows from~\eqref{equation:exact_density}, and the last equality follows from the ordering of the singular values.
\end{proof}

\begin{proof}[Proof of~\cref{thm:uniform}]
     For brevity, let 
     \begin{equation*}
         A_k(X, x) := \cbr{\Usvd{X}= \Usvd{x}, \Vsvd{X} = \Vsvd{x}, \sval{X}_{[-k]} = \sval{x}_{[-k]}, k \leq r(\sval{X})}.
     \end{equation*}
     We define
    \begin{align*}
        F_{x, k}(t)
        &:= \prob \rbr{
            \sval{X}_k \leq t~\big|~A_k(X, x)
        }
    \end{align*}
    % denote the CDF of $\sval{X}_k$ conditional on $A_k(X)$.
    where the probability is taken over $X$ under the model~\cref{equation:mvn_model} with mean $\Theta$. By~\cref{prop:theta_delta_k}, $F_{x, k}(\sval{X}_k) = 1 - \PP_{\pivot{x}_k, k}(X)$. 
    Thus, by the probability integral transform, for any $\alpha \in (0, 1)$,
    \begin{align*}
        \prob \left( \PP_{\pivot{x}_k,k}(X)
        \leq \alpha~|~A_k(X, x)\right)
        &= 1 - \prob \left( \PP_{\pivot{x}_k,k}(X)\geq \alpha~\bigg|~A_k(X, x)\right)
        \\
        &= 1 - \prob \left( 1 - \PP_{\pivot{x}_k,k}(X)\leq 1 - \alpha~\bigg|~A_k(X, x)\right)
        \\
        &= 1 - \prob \left( F_{x, k}(\sval{X}_k) \leq 1 - \alpha~\bigg|~A_k(X, x)\right) \\
        &= 1 - \prob \left( \sval{X}_k \leq F_{x, k}^{-1}\rbr{1 - \alpha}~\bigg|~A_k(X, x)\right) \\
        &=
        1 - F_{x, k}\rbr{F_{x, k}^{-1}(1-\alpha)}\\
        &= 1 - (1 - \alpha) \\
        &= \alpha.
    \end{align*}
    Applying the law of total expectations over $\Usvd{X}_{[-k]}, \Vsvd{X}_{[-k]}$, and $\sval{X}_{[-k]}$, recalling that we can write the probability as an expectation, we obtain the guarantee conditional on $\widetilde A_k(X, x)$ 
    in~\cref{eq:conditioning_event}. 
    We illustrate this argument in the general case. Let $\cbr{a(X) = a(x)}$ and $\cbr{b(X) = b(x)}$ be two events for some constant $x$ and random variable $X$. Suppose that a function $g$ satisfies $\EE[g(X)~|~a(X) = a(x), b(X) = b(x)] = \alpha$. Then for any distribution $H$ over $b(x)$,
    \begin{equation*}
        \int \EE[g(X)~|~a(X) = a(x), b(X) = b(x)] \, \mathrm{d}H(b(x)) = \int \alpha \, \mathrm{d}H(b(x)) = \alpha. 
    \end{equation*}
    Taking $H$ to be the conditional law of $b(X)~|~a(X) = a(x)$, by the law of total expectation
    \begin{equation*}
        \EE[g(X)~|~a(X) = a(x)] = \int \EE[g(X)~|~a(X) = a(x), b(X) = b(x)] \, \mathrm{d}H(b(x)).
    \end{equation*}
\end{proof} 

\begin{proof}[Proof of~\cref{thm:si_ci_coverage}]\label{proof:cor_si_ci_coverage}
    Define the set
    \begin{equation}
        \text{C}_{\alpha_1, k}\rbr{X^{(1)}} := \cbr{\delta \in \RR: \PP_{\delta, k}\rbr{X^{(1)}} \in \sbr{\alpha_1/2, 1-\alpha_1/2}}.
    \end{equation}
    To show that $\text{C}_{\alpha_1, k}\rbr{X^{(1)}}$ is an interval, we proceed as in Lemma A.1 of~\citet{lee_exact_2016}. Examining the conditional density in~\cref{equation:si_density}, we can see that it is a one-parameter exponential family, monotone in $\pivot{x^{(1)}}_k$ for a fixed $t$. Thus, by properties of such families, the CDF $F_{\PP_{\delta, k}\rbr{X^{(1)}}}(t)$ is monotone decreasing in $\delta$. By~\cref{prop:theta_delta_k}, $\PP_{\delta, k}\rbr{X^{(1)}} = 1 - F_{\PP_{\delta, k}\rbr{X^{(1)}}}(\sval{X^{(1)}}_k)$ conditional on $\widetilde A_k\rbr{X^{(1)}, x^{(1)}}$ and hence is monotone increasing in $\delta$ for every fixed $t$. It follows that $\PP_{\delta, k}\rbr{X^{(1)}}$ crosses $\alpha_1/2$ and $1 - \alpha_1/2$ at unique points in $\delta$, and hence we have the interval $\text{C}_{\alpha_1, k}\rbr{X^{(1)}} = \sbr{L_{\alpha_1/2}^{\mathrm{num}}\rbr{X^{(1)}}, U_{\alpha_1/2}^{\mathrm{num}}\rbr{X^{(1)}}}$ where $L_{\alpha_1/2}^{\mathrm{num}}\rbr{X^{(1)}}$ and $U_{\alpha_1/2}^{\mathrm{num}}\rbr{X^{(1)}}$ satisfy
    \begin{equation}
        \PP_{L_{\alpha_1/2}^{\mathrm{num}}\rbr{X^{(1)}}, k}\rbr{X^{(1)}} = \alpha_1/2,
        \quad\quad
        \PP_{U_{\alpha_1/2}^{\mathrm{num}}\rbr{X^{(1)}}, k}\rbr{X^{(1)}} = 1 - \alpha_1/2.
    \end{equation}
    By \cref{thm:uniform}, $\text{C}_{\alpha_1, k}\rbr{X^{(1)}}$ has the coverage guarantee
    \begin{align}
    \label{eq:pivot_ci_uniform}
        &\prob \left[\pivot{x^{(1)}}_k \in \text{C}_{\alpha_1, k}\rbr{X^{(1)}}~|~\widetilde A_k\rbr{X^{(1)}, x^{(1)}} \right]\nonumber\\
        &= \prob \left[ \PP_{\pivot{x^{(1)}}_k,\,k}\rbr{X^{(1)}} \in [\alpha_1/2, 1-\alpha_1/2]~|~\widetilde A_k\rbr{X^{(1)}, x^{(1)}} \right]\nonumber\\
        &= \prob \left[ \PP_{\pivot{x^{(1)}}_k,\,k}\rbr{X^{(1)}} \leq 1-\alpha_1/2~|~\widetilde A_k\rbr{X^{(1)}, x^{(1)}} \right]\nonumber\\
        &\quad -\prob \left[ \PP_{\pivot{x^{(1)}}_k,\,k}\rbr{X^{(1)}} \leq \alpha_1/2~|~\widetilde A_k\rbr{X^{(1)}, x^{(1)}} \right]\nonumber\\
        &= (1 - \alpha_1/2) - (\alpha_1 / 2)\nonumber\\
        &= 1 - \alpha_1.
    \end{align}
    To construct a confidence interval for $\rbr{\pivot{x^{(1)}}_k}^2$, we define $\sbr{\widetilde L_{\alpha_1}^{\mathrm{num}} \rbr{X^{(1)}}, \widetilde U_{\alpha_1}^{\mathrm{num}} \rbr{X^{(1)}}}$ where
    \begin{align*}
        \widetilde L_{\alpha_1}^{\mathrm{num}} \rbr{X^{(1)}} &:= \min\cbr{ \abr{L_{\alpha_1/2}^{\mathrm{num}}\rbr{X^{(1)}}} , \abr{U_{\alpha_1/2}^{\mathrm{num}}\rbr{X^{(1)}}} }^2 \II[\mathrm{sign}\rbr{L_{\alpha_1/2}^{\mathrm{num}}\rbr{X^{(1)}}  U_{\alpha_1/2}^{\mathrm{num}}\rbr{X^{(1)}}} = 1]
        \\
        \widetilde U_{\alpha_1}^{\mathrm{num}} \rbr{X^{(1)}} &:= \max\cbr{ \abr{L_{\alpha_1/2}^{\mathrm{num}}\rbr{X^{(1)}}} , \abr{U_{\alpha_1/2}^{\mathrm{num}}\rbr{X^{(1)}}} }^2.
    \end{align*}
    The definition of the minimum accounts for the fact that $\sbr{L_{\alpha_1}^{\mathrm{num}} \rbr{X^{(1)}}, U_{\alpha_1}^{\mathrm{num}} \rbr{X^{(1)}}}$ may cover zero; if zero is covered, the bounds have opposite signs and the lower bound of the confidence interval for the square must be zero.
    
    Now recall our confidence interval from~\cref{equation:si_confidence_interval} is
    \begin{equation*}
        % \text{CI}_{\alpha,k}\rbr{X^{(1)}, X^{(2)}} := \bigcup \cbr{\sbr{ \frac{ \delta^2 }{\lambda_{\alpha / 4}\rbr{X^{(2)}}},
        % \frac{ \delta^2 }{\lambda_{1 - \alpha / 4}\rbr{X^{(2)}}}
        % }~:~
        % \delta \in \RR, \delta \in \text{C}_{\alpha_1, k}\rbr{X^{(1)}}}.
        \text{CI}_{\alpha_1, \alpha_2, k}\rbr{X^{(1)}, X^{(2)}}
        := \sbr{ \frac{ \widetilde L_{\alpha_1}^{\mathrm{num}} \rbr{X^{(1)}} }{U_{\alpha_2}^{\mathrm{denom}} \rbr{X^{(2)}}},
        \frac{ \widetilde U_{\alpha_1}^{\mathrm{num}} \rbr{X^{(1)}} }{L_{\alpha_2}^{\mathrm{denom}}\rbr{X^{(2)}}}
        }.
    \end{equation*}
    Therefore,
    \begin{align*}
        & \Pr\left(\frac{\rbr{\pivot{x^{(1)}}_k}^2}{\|\Theta\|_F^2} \in \text{CI}_{\alpha_1, \alpha_2, k}\rbr{X^{(1)}, X^{(2)}} \mid \tilde{A}_k\rbr{X^{(1)}, x^{(1)}} \right)\\
        & \geq  \Pr\Bigg( \cbr{ 
            \rbr{\pivot{x^{(1)}}_k}^2 \in \sbr{\widetilde L_{\alpha_1}^{\mathrm{num}}\rbr{X^{(1)}}, \widetilde U_{\alpha_1}^{\mathrm{num}}\rbr{X^{(1)}}}
        }\\
        &\quad\quad\quad\quad\quad\cap \cbr{
            \nbr{\Theta}_F^2 \in [L_{\alpha_2}^{\mathrm{denom}}\rbr{X^{(2)}}, U_{\alpha_2}^{\mathrm{denom}}\rbr{X^{(2)}}]
        } \mid \tilde{A}_k\rbr{X^{(1)}, x^{(1)}} \Bigg) \\
        & \geq  \Pr\Bigg( \left\{ \Usvd{X^{(1)}}_k^\top \Theta \Vsvd{X^{(1)}}_k \in \text{C}_{\alpha_1, k}\rbr{X^{(1)}} \right\}\\
        &\quad\quad\quad\quad\quad\cap \left\{   \|\Theta\|_F^2 \in [L_{\alpha_2}^{\mathrm{denom}}\rbr{X^{(2)}}, U_{\alpha_2}^{\mathrm{denom}}\rbr{X^{(2)}}]    \right\} \mid \tilde{A}_k\rbr{X^{(1)}, x^{(1)}}  \Bigg) \\
        & =  \Pr\rbr{\pivot{x^{(1)}}_k \in \text{C}_{\alpha_1, k}\rbr{X^{(1)}}  \mid \tilde{A}_k\rbr{X^{(1)}, x^{(1)}} } \Pr\rbr{   \|\Theta\|_F^2 \in [L_{\alpha_2}^{\mathrm{denom}}\rbr{X^{(2)}}, U_{\alpha_2}^{\mathrm{denom}}\rbr{X^{(2)}}] } \\
        &\geq (1 - \alpha_1) (1 - \alpha_2)\\
        & \geq 1-\alpha_1 -\alpha_2.
    \label{eq:CI-ratio}
    \end{align*}
\end{proof}

\subsection{Proofs of Section 4 results}

\begin{proof}[Proof of~\cref{prop:elbow_solution_set}]
    We wish to characterize the set 
    \begin{equation*}
        R(x, k) := \cbr{t: k \leq r\left(\simpute{x}{t}^k\right)},
    \end{equation*}
    where the function $r(\cdot)$ is defined in~\eqref{equation:elbow_rule}. Recall that for our quantities to be well defined for all possible elbow-selected values $k \in \{1, \dots, p-2\}$, we have defined $\kappa_i(\sval{x}) = -\infty$ for $i \not \in \{2,\dots,p-1\}$. From~\eqref{equation:elbow_rule}, we immediately see that this set can be re-written as
     \begin{equation*}
        R(x, k) = \cbr{t: \max_{i \in \cbr{2, \dots, k}} \kappa_i\left(\simpute{x}{t}^k\right) \leq \max_{i \in \cbr{k+1, \dots, p-1}} \kappa_i\left(\simpute{x}{t}^k\right)},
    \end{equation*}
    since this is the event where the maximizer occurs at an index greater than $k$. To define $R(x, k)$ more explicitly, we define
    \begin{equation*}
        c_{1,k}(x) := 
        \begin{cases} 
          \underset{i \in \cbr{2, \dots, k-2}}{\max} \kappa_i\rbr{\sval{x}} & \text{if }k \geq 4 \\
          -\infty & \text{otherwise}
          \end{cases}
        \quad\text{and}\quad
        c_{2,k}(x) := \begin{cases} 
          \underset{i \in \cbr{k+2, \dots, p-1}}{\max} \kappa_i\rbr{\sval{x}} & \text{if } k \leq p-3 \\
          -\infty & \text{otherwise}
          \end{cases}
    \end{equation*}
    and, recalling the definition of $\simpute{x}{t}^k$ in~\cref{equation:s_impute}, rewrite it as
    \begin{equation*}
        R(x, k) = \cbr{t: \max\cbr{
            \kappa_{k-1}\rbr{\simpute{x}{t}^k}, 
            \kappa_{k}\rbr{\simpute{x}{t}^k}, 
            c_{1,k}(x) 
        }
        \leq
        \max\cbr{
            \kappa_{k+1}\rbr{\simpute{x}{t}^k},
            c_{2,k}(x) 
        }},
    \end{equation*}
    or equivalently as
    \begin{equation}
        \label{equation:two_part_t_solution}
        R(x, k) = A_k \cup B_k,
    \end{equation}
    where we define
    \begin{equation}\label{equation:AB}
    \begin{split}
        A_k &:= \cbr{t: \max\cbr{
                \kappa_{k-1}\rbr{\simpute{x}{t}^k},
                \kappa_{k}\rbr{\simpute{x}{t}^k},
                c_{1,k}(x) 
            }
            \leq
            \kappa_{k+1}\rbr{\simpute{x}{t}^k}
            }\\
        B_k &:= \cbr{t: \max\cbr{
                \kappa_{k-1}\rbr{\simpute{x}{t}^k},
                \kappa_{k}\rbr{\simpute{x}{t}^k},
                c_{1,k}(x) %:= \max_{i \in \cbr{2, \dots, k-2}} \kappa_i\rbr{\simpute{x}{t}^k}
            }
            \leq
            c_{2,k}(x)
            }.
    \end{split}
    \end{equation}
    How can we interpret $A_k$ and $B_k$? We can see that $A_k$ restricts the elbow from occurring at index $k+1$ while $B_k$ restricts the elbow from occurring at index $k+2$ or later. It is worth considering a couple of edge cases first.
    
    If $k=1$, then by definition $\kappa_{k-1}\rbr{\simpute{x}{t}^k}$, $\kappa_{k}\rbr{\simpute{x}{t}^k}$, and $c_{1,k}(x)$ are all negative infinity. Thus, since singular values are positive, $t$ can take any positive value. This is made obvious by the fact that we will always select an index greater than or equal to $1$ with the elbow rule.

    If $k=2$, then $\kappa_{k}\rbr{\simpute{x}{t}^k}$ is finite and contributes to a useful bound. The terms $\kappa_{k-1}\rbr{\simpute{x}{t}^k}$ and $c_{1,k}(x)$ are negative infinity and so, recalling the definition of the discrete second derivative in~\eqref{equation:discrete_deriv}, we can write
    \begin{equation}
        \begin{split}
            A_2 &:= \cbr{t:
                \kappa_{k}\rbr{\simpute{x}{t}^k}
            \leq
            \kappa_{k+1}\rbr{\simpute{x}{t}^k}
            }\\
            & = \cbr{t: \sval{x}_{k-1}^2 - 2t^2 + \sval{x}_{k+1}^2 \leq t^2 - 2 \sval{x}_{k+1}^2 + \sval{x}_{k+2}^2}\\
            & = \left[\frac13\rbr{\sval{x}_{k-1}^2 + 3\sval{x}_{k+1}^2 - \sval{x}_{k+2}^2}, \infty\right).
        \end{split}
    \end{equation}
    If $c_{2,k}(x)$ is finite, then $B_2$ is non-empty and we have
    \begin{equation}
        \begin{split}
        B_2 &:= \cbr{t:
                \kappa_{k}\rbr{\simpute{x}{t}^k}
            \leq
            c_{2,k}(x)}\\
            & = \cbr{t: \sval{x}_{k-1}^2 - 2t^2 + \sval{x}_{k+1}^2 \leq c_{2,k}(x)}\\
            & = \left[\frac12\rbr{\sval{x}_{k-1}^2 + \sval{x}_{k+1}^2 - c_{2,k}(x)}, \infty\right).
        \end{split}
    \end{equation}
    If $k \geq 3$, then the discrete derivatives of interest are all well defined and we can write $A_k$ in~\eqref{equation:AB} as the intersection of sets of solutions to the three inequalities:
    \begin{enumerate}[(i)]
        \item $\sval{x}_{k-2}^2 - 2\sval{x}_{k-1}^2 + t^2 = \kappa_{k-1}\rbr{\simpute{x}{t}^k} \leq \kappa_{k+1}\rbr{\simpute{x}{t}^k} = t^2 - 2 \sval{x}_{k+1}^2 + \sval{x}_{k+2}^2$,
        \item $\sval{x}_{k-1}^2 - 2t^2 + \sval{x}_{k+1}^2 = \kappa_{k}\rbr{\simpute{x}{t}^k} \leq \kappa_{k+1}\rbr{\simpute{x}{t}^k} = t^2 - 2 \sval{x}_{k+1}^2 + \sval{x}_{k+2}^2$,
        \item $c_{1,k}(x) \leq \kappa_{k+1}\rbr{\simpute{x}{t}^k} = t^2 - 2 \sval{x}_{k+1}^2 + \sval{x}_{k+2}^2$.
    \end{enumerate}
    Rearranging to solve for $t^2$ in each case ---- which is one-to-one with $t$ since singular values are non-negative by definition --- we first notice that the $t^2$ terms in (i) cancel, leaving us with the inequality
    \begin{equation}\label{equation:A_form}
        \sval{x}_{k-2}^2 - 2\sval{x}_{k-1}^2 \leq- 2 \sval{x}_{k+1}^2 + \sval{x}_{k+2}^2,
    \end{equation}
    which does not involve $t$. If this inequality does not hold, then $A_k$ in~\eqref{equation:AB} is the empty set. Otherwise, (ii) and (iii) can be combined into
    \begin{equation*}
        \max\cbr{\sval{x}_{k-1}^2 + 3\sval{x}_{k+1}^2 - \sval{x}_{k+2}^2, c_{1,k}(x) + 2\sval{x}_{k+1}^2 - \sval{x}_{k+2}^2}\leq t^2.
    \end{equation*}
    Note that in the event that $c_{1,k}(x) = -\infty$, only inequality (ii) will contribute to the lower bound of $t^2$.
    
    We likewise write $B_k$ in~\eqref{equation:AB} as the intersection of sets of solutions to the three inequalities:
    \begin{enumerate}[(i')]
        \item $c_{1,k}(x) \leq c_{2,k}(x)$ and $|c_{2,k}(x)|<\infty$.
        \item $\sval{x}_{k-2}^2 - 2\sval{x}_{k-1}^2 + t^2 = \kappa_{k-1}\rbr{\simpute{x}{t}^k} \leq c_{2,k}(x)$,
        \item $\sval{x}_{k-1}^2 - 2t^2 + \sval{x}_{k+1}^2 = \kappa_{k}\rbr{\simpute{x}{t}^k} \leq c_{2,k}(x)$.
    \end{enumerate}
    Similarly to before, if the inequality in (i') does not hold, then $B_k$ in~\eqref{equation:AB} is the empty set. Note that this also is true for the case $k=2$, as seen in the resulting proposition. Otherwise, (ii') and (iii') can be combined into
    \begin{equation}\label{equation:B_form}
        \frac12\rbr{\sval{x}_{k-1}^2 + \sval{x}_{k+1}^2 - c_{2,k}(x)} \leq t^2 \leq 2\sval{x}_{k-1}^2 - \sval{x}_{k-2}^2 + c_{2,k}(x).
    \end{equation}
    It follows that the set $R(x, k)$ is one or both of the intervals $A_k$ and $B_k$, depending on which of conditions (i) and (i') hold. From~\eqref{equation:A_form} and~\eqref{equation:B_form}, we can write $A_k$ and $B_k$ as explicit intervals as in the proposition statement.
    
    Finally, we note that at least one of (i) and (i') must hold. Otherwise, by \eqref{equation:AB}, $R(x, k)$ would be the union of two empty sets and hence itself the empty set. However, we know that $k \leq r(\sval{x})$ by definition of $k$ in the proposition.
      Therefore, $t = \sval{x}_k \in R(x, k)$, and thus the set $R(x, k)$ is non-empty. 
\end{proof}

\section{Additional simulation results}\label{sec:sim_results}
\renewcommand{\thefigure}{B\arabic{figure}}
\setcounter{figure}{0}

We provide extended results from the main text regarding the ZG rule, additional results using the discrete derivative elbow rule, and detail the simulation setups used for both rules.

\subsection{ZG rule}\label{sec:appendix_zg_rule}

The data generation procedure was outlined in~\cref{subsec:sim-datagen}. We choose the singular values  so that the elbow rule in question selects the true rank in the noiseless setting, i.e., $r(\sval{\Theta}) = \rank(\Theta)$. Under the ZG-rule, we set $\sval{\Theta} = (5, 4, 3, 2, 1, 0,\ldots,0)^{1/5} \times (np)^{1/4}$. This setup follows from the experimental setup of~\citet{choi_selecting_2017}.

\cref{fig:power_medium_rank} displays the results from~\cref{fig:zg-power}, but instead plotted against the signal-to-noise ratio $\pivot{X}_k / \sigma$.

\begin{figure}[!htb]
    \centering
    \includegraphics[width=\linewidth]{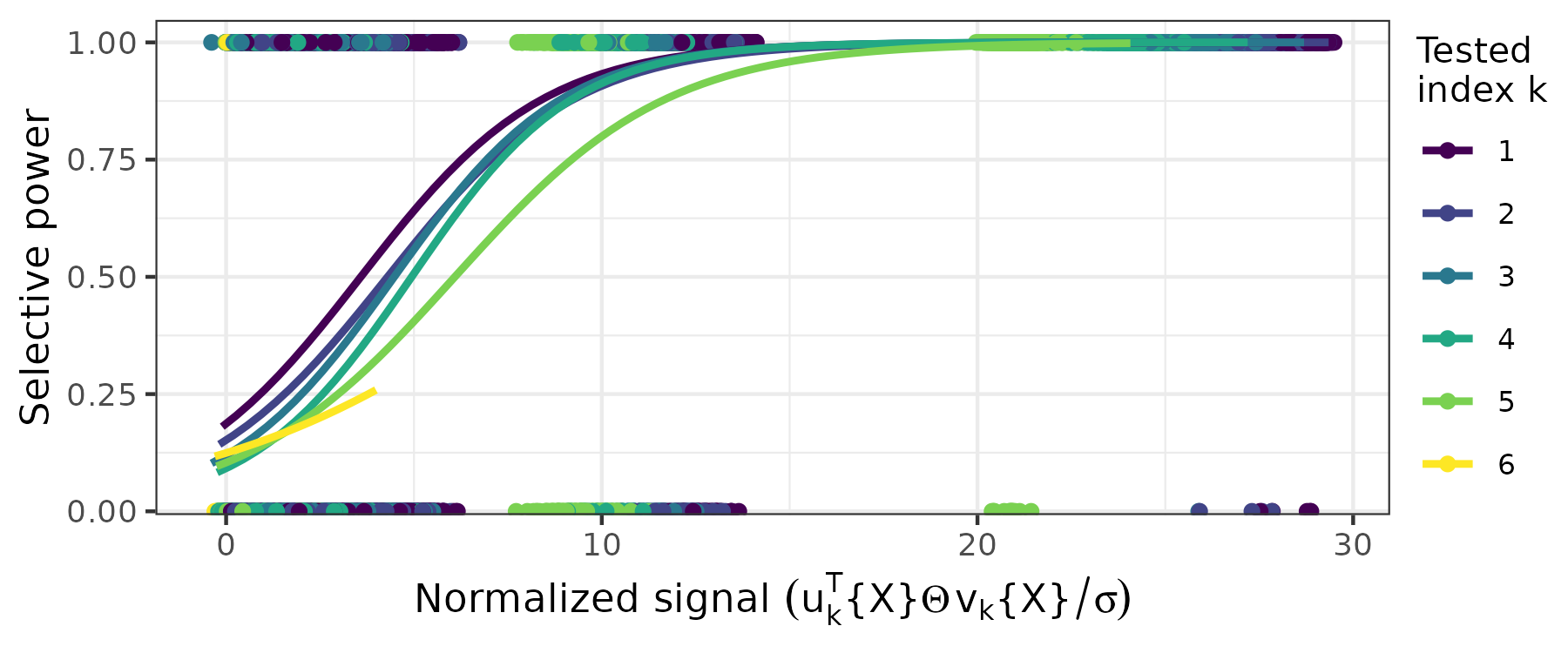}
    \caption{The simulation setting of~\cref{fig:zg-power}, but with power plotted versus the signal-to-noise ratio $\pivot{X}_k / \sigma$. Results are smoothed using a quasi-binomial GLM, and qualitatively agree with \cref{fig:zg-power}.} 
    \label{fig:power_medium_rank}
\end{figure}

In the following figures, we replicate the simulation results in the main text, but using an estimate of the variance described in \cref{sec:application}. \cref{fig:power_est_var} replicates \cref{fig:toy_example} and \cref{fig:zg-power}. \cref{fig:conf_coverage_var_est} replicates \cref{fig:zg-ci_coverage}. \cref{fig:conf_width_var_est} replicates \cref{fig:zg-ci_coverage_screeplot}. The takeaway is that using the variance estimate instead of the true variance preserves the validity of inference at the cost of slightly lower power.

\begin{figure}[!htb]
    \centering
    \includegraphics[width=0.43\linewidth, valign=t]{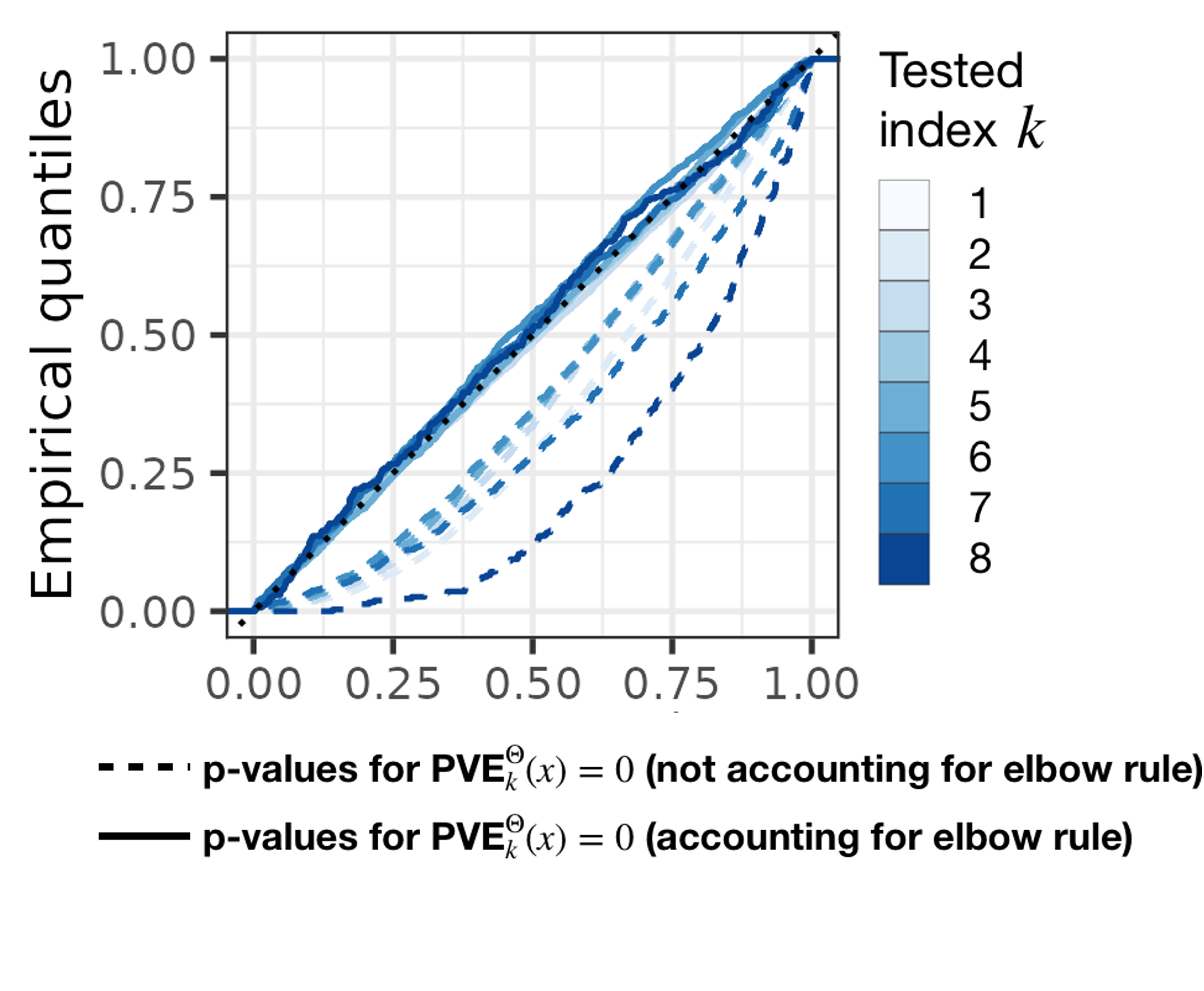}\includegraphics[width=0.57\linewidth, valign=t]{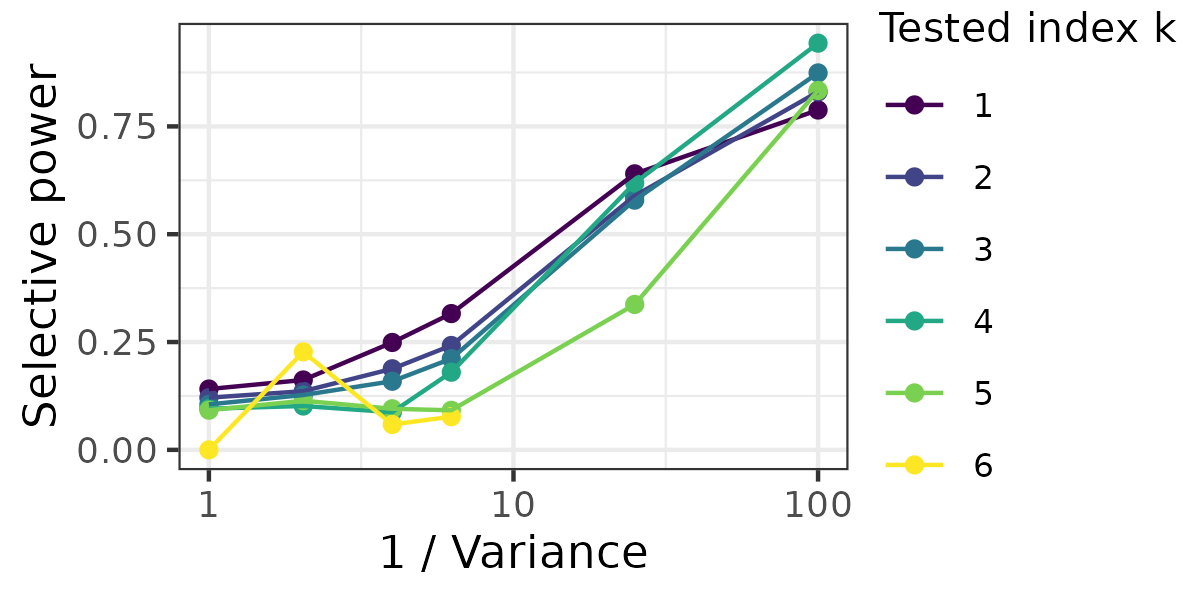}
    \caption{\textbf{(Left):} The simulation setting of~\cref{fig:toy_example}, but with $\sigma=0.1$ and an estimate of the variance instead of the true variance. The p-values of the test accounting for the elbow rule remain approximately uniformly distributed and thus properly control the Type 1 error.  \textbf{(Right):} The simulation setting of~\cref{fig:zg-power}, but with an esimate of the variance instead of true variance. This results in lower power.}
    \label{fig:power_est_var}
\end{figure}

\begin{figure}[!htb]
    \centering
    \includegraphics[width=\linewidth]{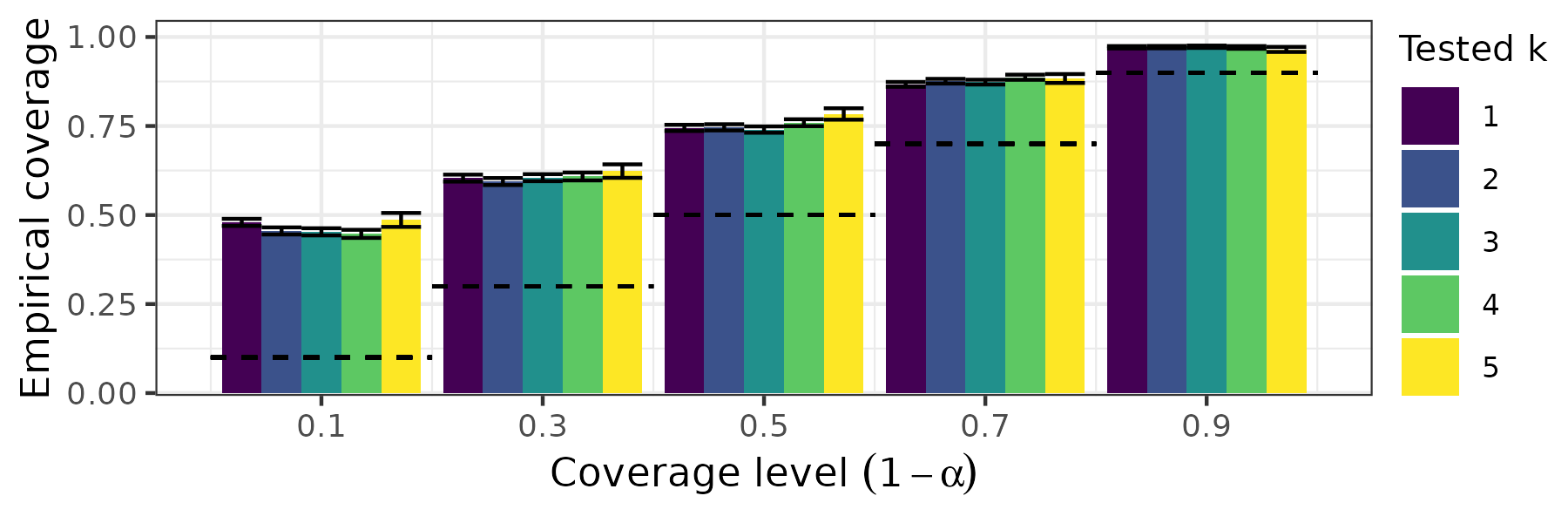}
    \caption{The simulation setting of \cref{fig:zg-ci_coverage}, but with an estimate of the variance instead of the true variance. The confidence intervals still attain the nominal coverage.}
    \label{fig:conf_coverage_var_est}
\end{figure}

\begin{figure}[!htb]
    \centering
    \includegraphics[width=\linewidth]{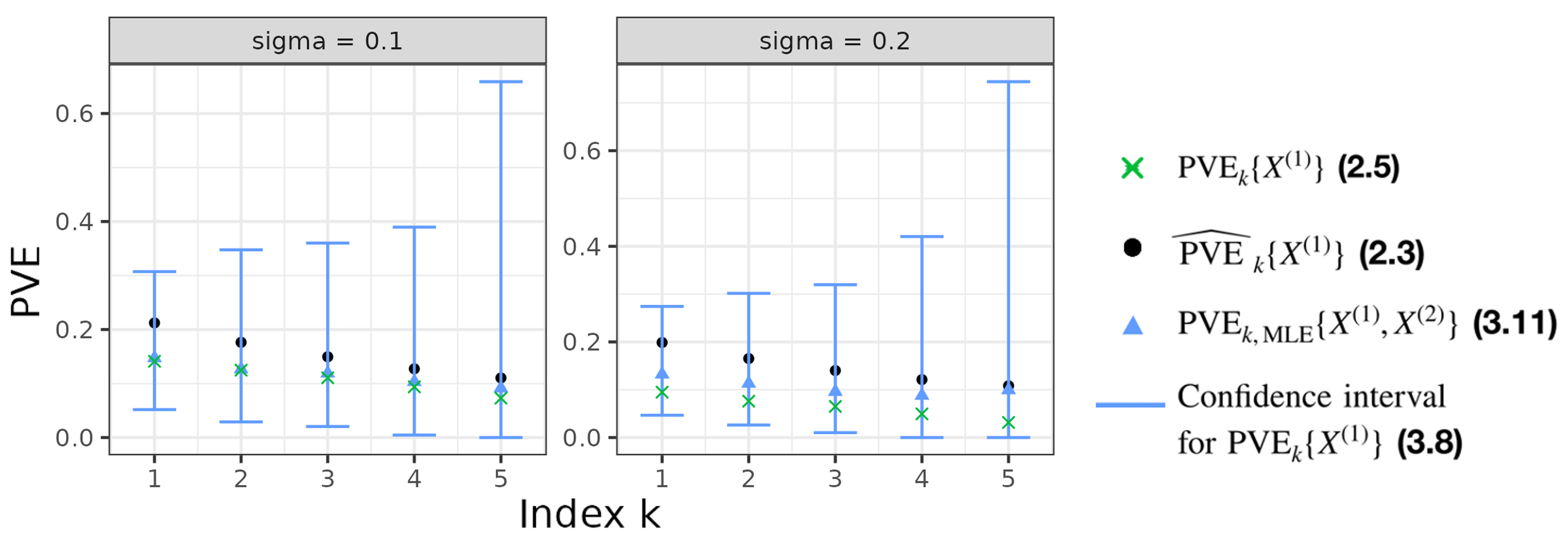}
    \caption{The simulation setting of \cref{fig:zg-ci_coverage_screeplot}, but using an estimate of the variance instead of the true variance. We see that the confidence intervals are wider when the variance is estimated.}
    \label{fig:conf_width_var_est}
\end{figure}

\subsection{Derivative rule}\label{sec:appendix_deriv_rule}

We provide simulation results here for hypothesis testing and confidence intervals under the derivative rule variant of the elbow-rule. Computation under this rule is analytically tractable, and results are qualitatively similar to those of the ZG rule.

We follow the data generation procedure outlined in~\cref{subsec:sim-datagen}. We choose the singular values under the alternative so that the elbow rule in question selects the true rank in the noiseless setting, i.e., $r(\sval{\Theta}) = \rank(\Theta)$. Under the derivative-rule, we set $\sval{\Theta} = (5, 4, 3, 2, 1, 0,\ldots,0)^{1/5} \times (np)^{1/4}$. This setup follows from the experimental setup of~\citet{choi_selecting_2017} and our choice of $n=50$ and $p=10$. Results are shown in Figures~\ref{fig:elbow-power}~-~\ref{fig:elbow-ci_coverage_screeplot}.

\begin{figure}[!htb]
    \centering
    \includegraphics{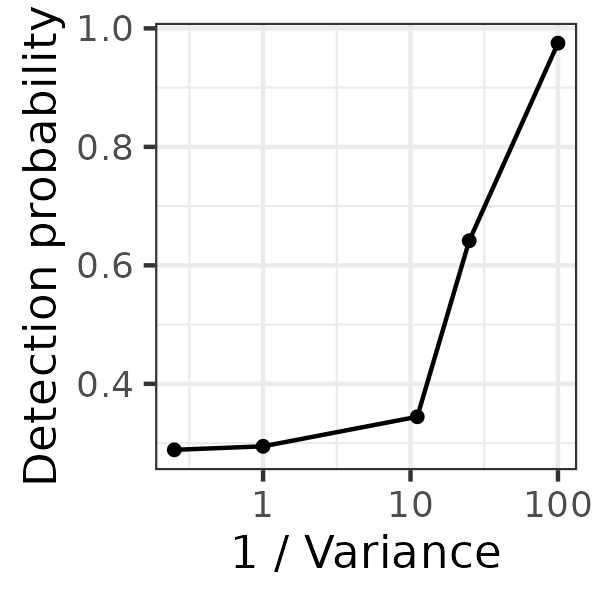}
    \includegraphics{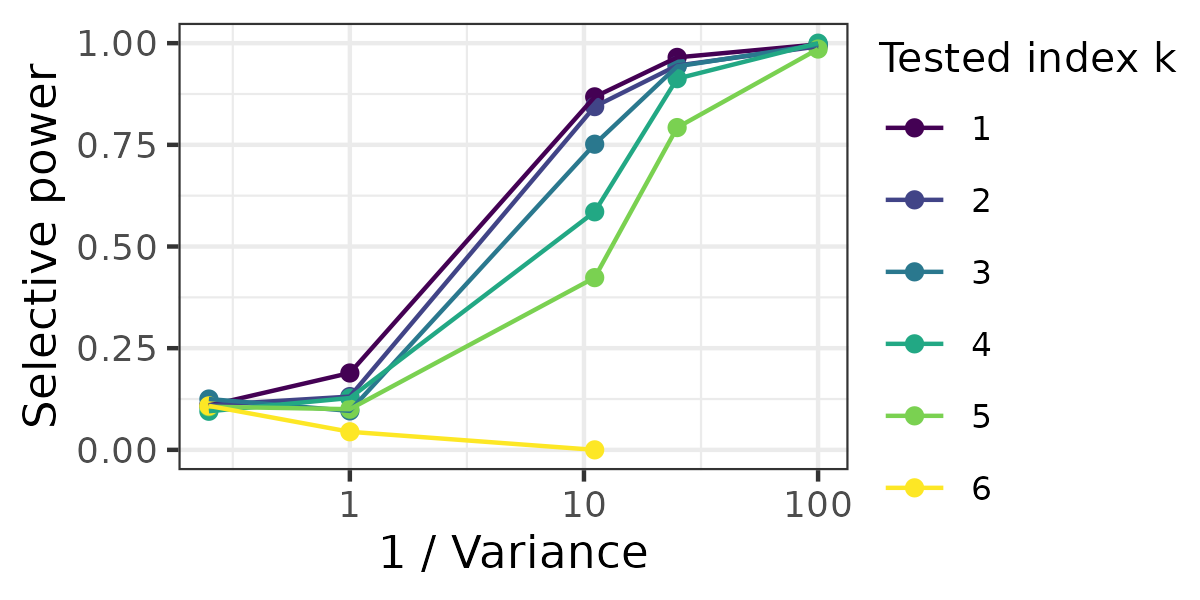}
    \caption{As in~\cref{fig:zg-power}, but with the discrete derivative elbow rule. The results are qualitatively the same.}
    \label{fig:elbow-power}
\end{figure}

\begin{figure}[!htb]
    \centering
    \includegraphics[width=\linewidth]{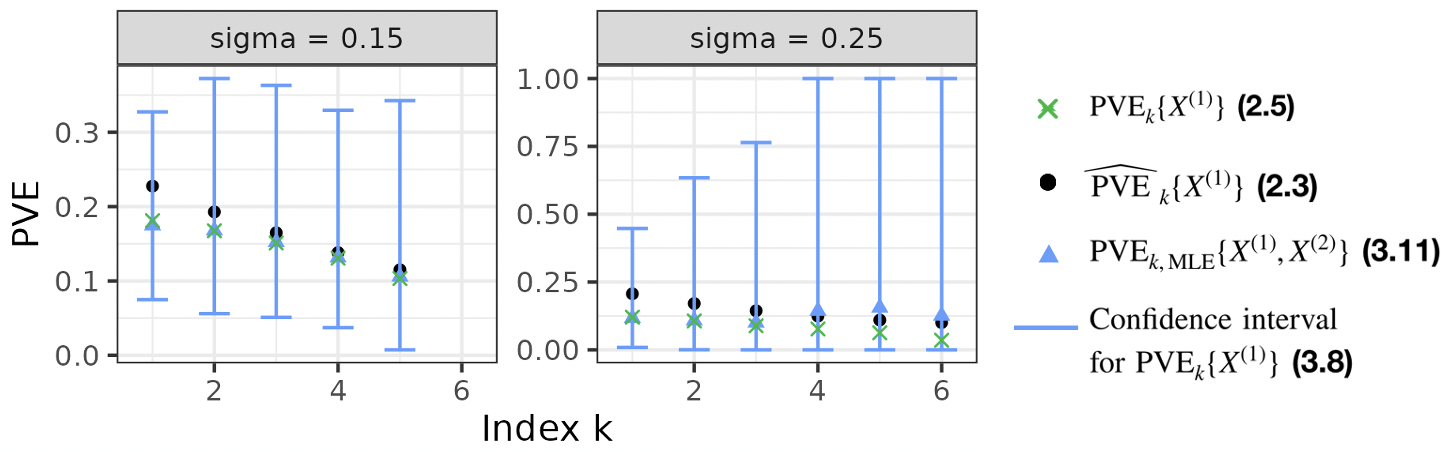}
    \caption{As in~\cref{fig:zg-ci_coverage_screeplot}, but with the discrete derivative elbow rule and $\sigma \in \cbr{0.15, 0.25}$. The results are qualitatively the same.}
    \label{fig:elbow-ci_coverage_screeplot}
\end{figure}

\end{document}